%% file: main.tex
\documentclass[11pt,a4paper]{article}
 
\usepackage{amsmath,amsthm,amssymb}
\usepackage{mathtools}
\mathtoolsset{showonlyrefs=true}
\allowdisplaybreaks

\usepackage{xspace}
\usepackage[english]{babel}
\usepackage{hyperref}
\usepackage{color}
\usepackage[T1]{fontenc}
\usepackage[utf8]{inputenc}
\usepackage{bbm}
\usepackage{constants}
\newconstantfamily{small}{
symbol=c,
}
\usepackage{tikz}
\usepackage[margin=1in]{geometry}
\usepackage[ruled]{algorithm2e}
\usepackage{bm}
\usepackage{tikz}

\theoremstyle{plain}
\newtheorem{theorem}{Theorem}[section]

\newtheorem{lemma}[theorem]{Lemma}
\newtheorem{corollary}[theorem]{Corollary}

\newtheorem{observation}[theorem]{Observation}

\theoremstyle{definition}
\newtheorem{definition}[theorem]{Definition}

\theoremstyle{remark}
\newtheorem{remark}{Remark}

\newcommand{\nfrac}[2]{#1/#2}
\newcommand{\weak}{weak reconstruction\xspace}
\newcommand{\strong}{strong reconstruction\xspace}

\newcommand{\planted}{\mathcal{G}_{2n,p,q}}

\newcommand{\Prob}[2]{\mathbf{P}_{#1} \left( #2 \right)}
\newcommand{\Expec}[2]{\mathbf{E}_{#1} \left[ #2 \right]}
\newcommand{\Var}[2]{\mathrm{Var}_{#1} \left( #2 \right)}
\newcommand{\Vr}[1]{\mathrm{Var} \left( #1 \right)}

\newcommand{\skproof}{\noindent\textit{Outline of Proof. }}
\newcommand{\ideaproof}{\noindent\textit{Idea of the proof. }}

\newcommand{\bigO}{\mathcal{O}}

\newcommand{\aveprot}{\mbox{\sc Averaging}}

\newcommand{\avg}{\mbox{\sc Averaging}\xspace}

\def\bx{{\bf x}}
\def\by{{\bf y}}
\def\bv{{\bf v}}
\def\bw{{\bf w}}
\def\bz{{\bf z}}
\def\be{{\bf e}}
\def\bchi{{\bm\chi}}
\def\bone{{\mathbf 1}}
\def\bzero{{\mathbf 0}}

\def\R{{\mathbb R}}

\let\pr=\Prob
\renewcommand{\Pr}{\Prob{}}
\def\E{\Expec{}}

\renewcommand{\leq}{\leqslant}

\renewcommand{\geq}{\geqslant}
\renewcommand{\epsilon}{\varepsilon}

\title{\textbf{Find Your Place: Simple Distributed Algorithms for Community Detection}}
\author{Luca Becchetti\\
    {\small{}Sapienza Università di Roma}\\
    {\small{} Rome, Italy}\\
    {\small{}becchetti@dis.uniroma1.it} 
    \and Andrea Clementi \\ 
    {\small{}Università di Roma Tor Vergata}\\
    {\small{} Rome, Italy}\\
    {\small{}clementi@mat.uniroma2.it} 
    \and Emanuele Natale\\
    {\small{}Sapienza Università di Roma}\\
    {\small{} Rome, Italy}\\
    {\small{}natale@di.uniroma1.it} 
    \and Francesco Pasquale \\
    {\small{}Sapienza Università di Roma}\\
    {\small{} Rome, Italy}\\
    {\small{}pasquale@dis.uniroma1.it} 
    \and Luca Trevisan \\
    {\small{}U.C. Berkeley}\\
    {\small{} Berkeley, CA, United States}\\
    {\small{}luca@berkeley.edu} 
}
\date{}

\begin{document}
\maketitle
\thispagestyle{empty}

\begin{abstract}
\input{./abstract}

\end{abstract}
\newpage
\setcounter{page}{1}

\input{./trunk/intro-lucab.tex}

\input{./trunk/preliminaries-full.tex}

\input{./trunk/regular-full.tex}

\input{./trunk/nonregular-full.tex}

\input{./trunk/tightblock.tex}

\input{./trunk/conclusions.tex}

\bibliographystyle{plain}
\bibliography{pbm}

\newpage
\appendix
\begin{center}
\LARGE{\textbf{Appendix}}
\end{center}

\input{./trunk/apx-preli.tex}

\input{./trunk/apx-wc.tex}
\input{./trunk/apx-random.tex}

\input{./trunk/apx-more.tex}

\end{document}

%% file: abstract.tex
Given an underlying graph, we consider the following \emph{dynamics}: 
Initially, each node locally chooses a value in $\{-1,1\}$, 
uniformly at random and independently of other nodes. Then, in each consecutive 
round, every node updates its local value to the average of the values 
held by its neighbors, at the same time applying an elementary, local clustering rule that only 
depends on the current and the previous values held by the node. 


We prove that the process resulting from this dynamics produces a 
clustering that exactly or approximately (depending on the graph) 
reflects the underlying cut in logarithmic time, under various 
graph models that exhibit a sparse balanced cut, including the 
stochastic block model.  
We also prove that a natural extension of this 
dynamics performs community detection on a regularized
version of the stochastic block model with multiple communities.

Rather surprisingly, our results provide rigorous evidence for the ability of
an extremely simple and natural dynamics to address a computational problem that is non-trivial even in a
centralized setting. 


%% file: trunk/intro-lucab.tex
\newpage

\section{Introduction} 

Consider the following distributed algorithm on an undirected 
graph: At the  outset, every node picks an initial value, 
independently and uniformly at random in $\{ -1,1\}$; then, in each 
synchronous round, every node updates its value to the average of 
those held by its neighbors. A node also tags itself ``blue'' if the last update 
increased its value, ``red'' otherwise.  


We show that under various graph models exhibiting sparse balanced cuts, including the 
\emph{stochastic block model}~\cite{holland1983stochastic}, the process resulting from the above simple 
local rule converges, in logarithmic time, to a colouring that exactly or approximately (depending on 
the model) reflects the underlying cut.  
We further show that our approach simply and naturally extends to more  
communities, providing a quantitative analysis for a 
regularized version of the stochastic block model with multiple communities. 

Our algorithm is one of the few examples of a
\emph{dynamics}~\cite{AAE07,afek2011biological,Doty14,MNT14} that 
solves  a computational problem
that is non-trivial in a centralized setting.
By {\em dynamics} we here mean synchronous distributed algorithms
characterized by a very simple structure, whereby the state of a node at 
round $t$ depends only on its state and a symmetric function of the  
multiset of states of its neighbors at round $t-1$, while the update rule is 
the same for every graph and every node and does not change over time. Note that 
this definition implies that the network is {\em anonymous}, that is, 
nodes  do not possess distinguished identities.
Examples of dynamics include update rules in which every 
node updates its state to the plurality or the median of the states of 
its neighbors,\footnote{When states correspond 
to rational values.} or, as is the case in this 
paper, every node holds a value, which it updates to the average of the 
values held by its neighbors. In contrast, an algorithm 
that, say, proceeds in two phases, using  averaging during the first 
$10\log n$ rounds and plurality from round  $1+ 10\log n$ onward, with $n$ 
the number of nodes, is not a dynamics according 
to our definition, since its update rule depends on the size of 
the graph. As another example, an algorithm that starts by having the lexicographically first 
vertex elected as  ``leader'' and then propagates its state to 
all other nodes again does not meet our definition of dynamics, since 
it assigns roles to the nodes and requires them to possess 
distinguishable identities.

The \avg dynamics, in which each node updates its value 
to the average of its neighbors, is perhaps one of the simplest and most 
interesting examples of linear dynamics and it always converges when $G$ is connected and not 
bipartite: It converges to the global average of the initial values 
if the graph is regular and 
to a weighted global average if it isn't~\cite{boyd_randomized_2006,shah2009gossip}.
Important  applications of linear dynamics have been proposed in 
the recent past 
\cite{kempe_gossip-based_2003,aysal2009broadcast,tsitsiklis1984problems,kleinberg1999authoritative},
for example to  perform basic tasks such as 
self-stabilizing \emph{consensus} in  faulty distributed systems \cite{benezit2009interval,xiao2007distributed,olshevsky2009convergence}.
The convergence time of the \avg dynamics is the mixing time of a random
walk on $G$ \cite{shah2009gossip}. It is logarithmic in $|V|$ if the underlying 
graph is a  good \emph{expander} \cite{hoory2006expander}, while it is slower 
on graphs that exhibit sparse cuts.

While previous work on applications of linear dynamics has focused on 
tasks that are specific to distributed computing (such as reaching 
consensus, or stability in the presence of faulty nodes), in this paper 
we show that our simple protocol based on the  
the \avg\ dynamics is able to address community detection, i.e.,  
it identifies partition $(V_1,V_2)$ of a clustered 
graph $G=((V_1,V_2),E)$, either exactly (in 
which case we have a {\em strong} reconstruction algorithm) or
approximately (in which case we speak of a {\em weak} reconstruction algorithm).

\subsection{Our contributions}
Consider a graph $G=(V,E)$. We show that, if a partition 
$(V_1,V_2)$ of $G$ exists, such that
$\bone_{V_1} - \bone_{V_2}$ is\footnote{As explained further, $\bone_{V_i}$, is the vector with 
$|V|$ components, such that the $j$-th component is $1$ if $j\in 
V_i$, it is $0$ otherwise.} (or is close to) a right-eigenvector 
of the second largest eigenvalue of the transition matrix of $G$, and 
the gap between the second and the third largest eigenvalues is 
sufficiently large, our algorithm identifies the partition 
$(V_1,V_2)$, or a close approximation thereof, in a logarithmic 
number of rounds.
Though the algorithm we propose does not explicitly perform any 
eigenvector computation, its exploits the spectral structure of the 
underlying graph, based on the intuition that the dynamics is a 
distributed simulation of the power method. Our analysis involves
two main novelties, relating to how nodes assign themselves to clusters, and
to the spectral bounds that we prove for certain classes of graphs.

A conceptual contribution is to make each node, at each round $t$, 
assign itself to a cluster (``find its place'') by
considering the difference between its value at time $t$ and its value at time $t-1$.
Such a  criterion removes the component of the value lying in the first eigenspace
without explicitly computing it. This idea has two advantages: it allows a
particularly simple algorithm, and it gives a running time that depends
on the third eigenvalue of the transition matrix of the graph. In graphs
that have the structure of two expanders joined by a sparse cut, the running
time of the algorithm depends only on  the expansion of the components and
it is faster than the mixing time of the overall graph. To the best of our knowledge, this is the first 
distributed block reconstruction algorithm  converging faster than the mixing 
time. 

Our algorithm works on any graph where (i) the right-eigenspace of the second 
eigenvalue of the transition matrix is correlated to the cut between the two clusters and (ii) the gap 
between the second and third eigenvalues is sufficiently large.
While these conditions have been investigated for the spectrum of the \textit{adjacency}
matrix of the graph, our analysis requires these conditions to hold for 
the \textit{transition} matrix. A technical contribution of this paper is to show that such  
conditions are met by a large  class of graphs, that includes graphs sampled  
from the \textit{stochastic block model}. Proving spectral properties of the
transition matrix of a random graph is more challenging than proving
such properties for the adjacency matrix, because the entries of the transition 
matrix are not independent. 

\vspace{-12pt}
\paragraph{Strong reconstruction for regular clustered graphs.}


\noindent
A $(2n,d,b)$-clustered regular graph $G = ((V_1,V_2),E)$ is a connected
graph over vertex set $V_1 \cup V_2$, with 
$|V_1|=|V_2|=n$, adjacency matrix $A$, and such that every node has degree $d$ and 
it has (exactly) $b$ neighbors outside its cluster.    
If the two subgraphs induced by $V_1$ and $V_2$ are good 
expanders and $b$ is sufficiently small, the second and third 
eigenvalues of the graph's transition 
matrix $P = (1/d) \cdot A$ are separated by a large gap. In this case,  
we can prove that the following happens with high 
probability (for short \emph{w.h.p}\footnote{We say that a sequence  of  events 
$\mathcal E_n$,  $n= 1,2, \ldots$ holds \emph{with high 
probability} if  $ \Prob{}{\mathcal E_n} = 1- \bigO (1/n^{\gamma})$ for  
some positive constant $\gamma >0$.}): If the \avg dynamics is 
initialized by having every node choose a value uniformly and independently at random in 
$\{-1,1\}$, within a logarithmic number of rounds the system enters a regime in which
nodes' values are monotonically increasing or decreasing, depending on
the community they belong to. 
As a consequence, every node can apply a simple and completely local clustering 
rule in each round, which eventually results in a strong reconstruction (Theorem~\ref{thm:reg_stop}).

We then show that, under mild assumptions, a graph selected from the \textit{regular 
stochastic block model} \cite{brito_recovery_2015}  is  a 
$(2n,d,b)$-clustered regular graph  that satisfies  the above spectral gap hypothesis, w.h.p. 
We thus  obtain a fast and extremely simple dynamics for strong reconstruction, over the full range of parameters 
of the regular stochastic block model for which this is known to be possible using centralized algorithms
\cite{mossel_reconstruction_2014,brito_recovery_2015} (Section \ref{ssec:related} and 
Corollary~\ref{cor.regular}).

We further show that a natural 
extension of our algorithm, in which nodes maintain an array of values and an 
array of colors, correctly identifies a hidden balanced $k$-partition in a regular 
clustered graph with a gap between $\lambda_k$ and 
$\lambda_{k+1}$. Graphs sampled from the regular stochastic block model with 
$k$ communities satisfy such conditions,  w.h.p. (Theorem \ref{thm:more}).

\vspace{-12pt}
\paragraph{Weak reconstruction for non-regular clustered graphs.}

\noindent
As a  main technical contribution, we extend the above  analysis  to show that 
our dynamics  also ensures weak reconstruction in clustered graphs having 
two clusters that 
satisfy an approximate regularity condition and a gap between second 
and third eigenvalues of the transition matrix $P$ (Theorem~\ref{thm:irreg_stop}). 
As an application, we then prove that these conditions are
met by the \emph{stochastic block model}~\cite{abbe2014exact,coja2010graph,decelle_asymptotic_2011,
dyer1989solution,holland1983stochastic,jerrum_metropolis_1998,mcsherry2001spectral}, a 
random graph model that offers a popular framework for the 
probabilistic modelling of graphs that exhibit good clustering or 
community properties. We here consider  its  simplest    version, i.e.,   the  random 
graph $\planted$  consisting of $2n$ nodes and an edge probability 
distribution defined as follows: The node set is partitioned into
two subsets $V_1$ and $V_2$, each of size $n$; edges  linking 
nodes belonging to the same partition appear in $E$  independently at random with 
probability $p=p(n)$, while edges connecting nodes from different 
partitions appear with probability  $q=q(n)<p$.
Calling $a= pn$ and $b= qn$, we   prove  that  graphs sampled from $\planted$ satisfy the above 
approximate regularity and spectral gap conditions, w.h.p., whenever 
$a-b > \sqrt{(a+b) \cdot \log n}$ (Lemma~\ref{lemma:irreg}).

We remark that the latter result for the stochastic block model follows from an 
analysis that applies to general \emph{non-random} clustered graphs and 
hence does not exploit crucial properties of random graphs.
A further technical contribution of this paper is a refined, ad-hoc analysis of
the \avg dynamics for the $\planted$ model, showing that this protocol achieves weak-reconstruction, 
correctly classifying a $1-\epsilon$ fraction of vertices,  in logarithmic time
whenever $a-b > \Omega_\epsilon (\sqrt{(a+b)})$ and the expected degree 
$d = a+b$ grows at least logarithmically (Theorem~\ref{thm:tight_bounds}). 
This refined analysis requires a deeper understanding of the eigenvectors of 
the \textit{transition matrix} of $G$.
Coja-Oghlan~\cite{coja2010graph} defined certain graph properties 
that guarantee that a near-optimal bisection can be found based on 
eigenvector computations of the \emph{adjacency matrix}. Similarly, we show 
simple sufficient conditions under which a right eigenvector of the second largest 
eigenvalue of the transition matrix of a graph approximately 
identifies the hidden partition.  We give a tight analysis 
of the spectrum of the transition matrix of graphs sampled from the 
stochastic block model in Section~\ref{ssec:thm_irreg}. Notice that 
the analysis of the transition matrix is somewhat harder than that of the adjacency matrix, 
since the entries are not independent of each other; we were not 
able to find comparable results in the existing literature.

\subsection{Related work and additional remarks} \label{ssec:related}

\paragraph{Dynamics for block reconstruction.}
Dynamics received 
considerable attention in the recent past across different research 
communities, both as efficient   distributed 
algorithms~\cite{AAE07,benezit2009interval,olshevsky2009convergence,metivier2011optimal}
and as abstract models of \emph{natural}  interaction mechanisms 
inducing emergent behavior in complex systems \cite{afek2011biological,cardelli2012cell,Doty14,HouseHunt,MNT14}. 
For instance, simple averaging dynamics have 
been considered to model opinion formation mechanisms
\cite{degroot1974reaching,doi:10.1080/0022250X.1990.9990069}, while 
a number of other dynamics have been proposed to describe different
social phenomena \cite{easley2010networks}. 
\noindent
\emph{Label propagation algorithms}
\cite{raghavan2007near} are dynamics based on majority updating rules~\cite{AAE07} 
and have been applied to some
computational problems including clustering.
Several papers present experimental results  for
such protocols on specific classes of clustered graphs \cite{barber2009detecting,liu2010advanced,raghavan2007near}.
The only available rigorous analysis of label 
propagation algorithms on planted partition graphs  is the one presented in 
\cite{kothapalli2013analysis}, where the authors propose and analyze  a label propagation protocol
on   $\planted$ for dense topologies.
In particular, their analysis considers  the case where 
$p = \Omega(1/n^{1/4-\epsilon})$ and $q = \bigO(p^2)$, a parameter 
range in which very dense clusters of constant diameter separated by a 
sparse cut occur w.h.p.
In this setting, characterized by  a polynomial gap between $p$ and $q$,
simple combinatorial and concentration 
arguments show that the protocol converges in constant expected 
time. They also conjecture a logarithmic bound for sparser topologies.

Because of their relevance for the reconstruction problem, we also mention
another class of algorithms, \emph{belief propagation algorithms}, whose
simplicity is close to that of dynamics. 
\emph{Belief propagation algorithms}  are best known as
message-passing algorithms for performing inference in graphical models
\cite{MAC03}. 
Belief propagation cannot be   considered a dynamics: At each round, each node sends a
different message to each neighbors, thus  the  update rule is not symmetric
w.r.t. the neighbors,  requiring thus  port numbering \cite{SUO13},  and the   required  local memory
  grows linearly in the degree of the node. 
Non-rigorous methods have given strong evidence that some \emph{belief
propagation algorithms} are optimal for the reconstruction problem
\cite{decelle_asymptotic_2011}.
Its rigorous analysis is a major challenge; in particular, the convergence to
the correct value of belief propagation is far from being fully-understood on
graphs which are not trees \cite{WEI00, MK07}.
As we discuss in the next subsection,  more complex algorithms,  
inspired  by belief propagation,   have been rigorously shown to perform
reconstruction optimally.  

\vspace{-12pt}
\paragraph{General algorithms for block reconstruction.}
While an important goal, improving performance of spectral clustering 
algorithms and testing their limits to the purpose of block 
reconstruction is not the main driver behind this work. 
Still, for the sake of completeness, we next compare  our dynamics  to previous general  algorithms 
for block reconstruction.   


Several algorithms for community detection are  \textit{spectral}: They typically 
consider the eigenvector associated to the second eigenvalue of the adjacency 
matrix $A$ of $G$, or the eigenvector corresponding to the largest eigenvalue 
of the matrix $A- \frac dn J$~\cite{boppana1987eigenvalues,coja-oghlan_spectral_2005,coja2010graph,mcsherry2001spectral}
\footnote{$A$ is the adjacency matrix of $G$, $J$ is the matrix having all entries equal to $1$, $d$ is the average degree and $n$ is the number
of vertices.}, since these are correlated with the hidden partition.  
More recently spectral algorithms have been 
proposed~\cite{abbe2015detection,coja2010graph,mossel_proof_2013,krzakala2013spectral,BLM15} 
that find a weak reconstruction even in the sparse, tight regime. 

Even though the above mentioned algorithms have been presented in a centralized setting,
spectral algorithms turn out to be 
a feasible approach also for distributed models. Indeed,
Kempe and McSherry~\cite{kempe2004decentralized} show that eigenvalue 
computations can be performed in a distributed
fashion, yielding distributed algorithms for community detection in various 
models, including the stochastic block model.
However, the algorithm of Kempe and McSherry as well as any distributed version of the
above mentioned centralized algorithms are not dynamics. Actually, adopting the   effective 
concept from   Hassin and Peleg in \cite{hassin2001distributed}, such algorithms are even
not  \emph{light-weight}:
Different and not-simple operations are executed at different rounds, nodes have identities, 
messages are treated differently depending on the originator, and so on.   
Moreover, a crucial aspect is convergence time: 
The mixing time of the simple random walk on the graph is a bottleneck for the 
distributed algorithm of Kempe and McSherry and 
for any algorithm that performs community detection in a graph $G$ 
by employing the power method or the Lanczos method~\cite{lanczos1950iteration} as a subroutine to 
compute the eigenvector associated to the second eigenvalue of the 
adjacency matrix of $G$. Notice that the mixing time of  graphs sampled from 
$\planted$ is at least of the order of $\frac{a+b}{2b}$: hence, it  can be 
super-logarithmic and even $n^{\Omega(1)}$.

 

\medskip

In  general, the reconstruction problem has been studied extensively 
using a multiplicity of techniques, which include combinatorial algorithms  
\cite{dyer1989solution}, belief propagation \cite{decelle_asymptotic_2011}, spectral-based techniques \cite{mcsherry2001spectral,coja2010graph},  Metropolis approaches \cite{jerrum_metropolis_1998}, and semidefinite programming \cite{abbe2014exact}, among others.
Stochastic Block Models have been studied in a number of areas, 
including computer science 
\cite{boppana1987eigenvalues,mcsherry2001spectral,massoulie_community_2014}, 
probability theory~\cite{mossel_reconstruction_2014}, 
statistical physics \cite{decelle_asymptotic_2011}, and social 
sciences \cite{holland1983stochastic}. 
 Unlike the distributed setting, where the existence  of  
\emph{light-weight   protocols} \cite{hassin2001distributed} is  the main issue (even in non-sparse regimes), 
  in centralized setting  strong attention has been devoted
to establishing sharp thresholds for weak and strong reconstruction. 
Define $a= np$  as the expected {\em internal degree} (the number of 
neighbors  that each node has on the same side of the partition) and   
$b= nq$ as the expected {\em external degree} (the number of neighbors 
that each node has on the opposite side of the partition). Decelle et 
al. \cite{decelle_asymptotic_2011} conjectured that weak reconstruction 
is possible if and only if  $a-b> 2\sqrt{a + b}$. This was proved by 
Massoulie and  Mossel et al. 
\cite{mossel_proof_2013,massoulie_community_2014,mossel_reconstruction_2014}.
Strong recovery is instead possible if and only if $a-b > 2\sqrt{a + b} +
\log n$  \cite{abbe2014exact}.

Versions of the stochastic block model in which the 
random graph is regular have also been 
considered~\cite{mossel_reconstruction_2014,brito_recovery_2015}. 
In particular Brito et al.~\cite{brito_recovery_2015} show that strong 
reconstruction is possible in polynomial-time when $a-b > 2\sqrt{a+b-1}$.


%% file: trunk/preliminaries-full.tex
\section{Preliminaries} \label{sec:prely}
\paragraph{Distributed block reconstruction.} 
Let $G = \left((V_1, V_2), E\right)$ be a graph with $V_1 \cap V_2 = \emptyset$.
A weak (block) reconstruction is a two-coloring of the nodes that separates $V_1$ and $V_2$
up to a small fraction of the nodes. Formally, we define an $\epsilon$-\textit{weak 
reconstruction} as a map 
$f \,:\, V_1 \cup V_2 \rightarrow \{ \texttt{red}, \texttt{blue} \}$ such that
there are two subsets $W_1 \subseteq V_1$ and $W_2 \subseteq V_2$ with 
$|W_1 \cup W_2| \geqslant (1 - \epsilon) |V_1 \cup V_2|$
and $f(W_1) \cap f(W_2) = \emptyset$. When $\epsilon = 0$ we say that $f$ is a 
\textit{strong reconstruction}.
Given a graph $G = \left((V_1, V_2), E\right)$, 
the block reconstruction problem requires computing an $\epsilon$-reconstruction  of $G$.

In this paper, we propose the following distributed protocol. It is based on the averaging dynamics
and produces a coloring of the nodes at the end of every round. In the next
two sections we show that, within $\bigO(\log n)$ rounds, the coloring computed by the algorithm we propose 
achieves \emph{strong reconstruction} of the two blocks in the case 
of clustered regular graphs and \emph{weak 
reconstruction} in the case of clustered non-regular graphs. 

\begin{center}
\fbox{
\begin{minipage}{15cm}
\avg protocol:
\begin{description}
\item[Rademacher initialization:] At round $t = 0$ every node $v \in V$ 
independently samples its value from $\{- 1, +1\}$ uniformly at random;
\item[Updating rule:] At each subsequent round $t \geq 1$, every node $v \in V$
\begin{enumerate} 
\item (\avg dynamics) Updates its value $\bx^{(t)}(v)$ to the average of the 
values of its neighbors at the end of the previous round
\item (Coloring) If $\bx^{(t)}(v) \geqslant \bx^{(t-1)}(v)$ then $v$ sets 
$\texttt{color}^{(t)}(v) = \texttt{blue}$
otherwise $v$ sets $\texttt{color}^{(t)}(v) = \texttt{red}$.  
\end{enumerate}
\end{description}
\end{minipage}
}
\end{center}

The choice of the above coloring rule will be clarified 
in the next section, just before Theorem \ref{thm:reg_stop}. We give 
here  two remarks. First of all, the algorithm is completely 
oblivious to time, being a dynamics in the strictest sense. Namely, 
after initialization the protocol iterates over and over at every 
node. Convergence to a (possibly weak) reconstruction is a property 
of the protocol, of which nodes are not aware, it is something that 
eventually occurs. Second, the clustering criterion is completely {
\em local}, in the sense that a decision is individually and 
independently made by each node in each round, only on the basis of 
its state in the current and previous rounds. This may seem 
counterintuitive at first, but it is only superficially so. Despite 
being local, the clustering criterion uses information that reflects 
the global structure of the network, since nodes' values are related 
to the second eigenvector of the network's transition matrix. 
  
\paragraph{The \avg dynamics and random walks on $G$.}
The analysis of the \avg dynamics on a graph $G$ is closely related to the 
behavior of random walks in $G$, which are best studied using tools from 
linear algebra that we briefly summarize below.

Let $G=(V,E)$ be an undirected graph (possibly with multiple edges and self 
loops),  $A$ its adjacency matrix and $d_i$ the degree of node 
$i$. The {\em transition matrix} of (the random walk on) $G$ is the matrix 
$P = D^{-1} A$, where $D$ is the diagonal matrix such that $D_{i,i} = d_i$.
$P_{i,j} = (1/d_i) \cdot A_{i,j}$ is thus
the probability of going from $i$ to $j$ in one-step of the random walk on $G$. $P$
operates as the random walk process on $G$ by left multiplication, and as the \avg dynamics by right multiplication.
For $i = 1, 2$, define $\bone_{V_i}$, as the $|V|$-dimensional 
vector, whose $j$-th component is $1$ if $j\in 
V_i$, it is $0$ otherwise.
If $(V_1, V_2)$ is a bipartition of the nodes with $|V_1| = |V_2| = n$, 
we define the \textit{partition indicator vector} $\bchi = \bone_{V_1} 
- \bone_{V_2}$.
If $\bx$ is the initial vector of values, after $t$ rounds of the \avg dynamics
the vector of values at time $t$ is $\bx^{(t)} = P^t \bx$.
The product of the power of a matrix times a vector is best understood in 
terms of the spectrum of the matrix, which is what we explore in the next section. 

In what follows we always denote by $\lambda_1 \geq \ldots \geq 
\lambda_{2n}$ the eigenvalues of $P$. Recall that, since $P$ is a stochastic matrix
we have $\lambda_1 = 1$ and $\lambda_{2n} \geqslant -1$, moreover for all 
graphs that are connected and not bipartite it holds that $\lambda_2 < 1$ and 
$\lambda_{2n} > -1$. We denote by $\lambda$ the largest, in
absolute value, among all but the first two eigenvalues, namely
$\lambda = \max\left\{ | \lambda_i | \,:\, i = 3, 4, \dots, 2n \right\}$.
Unless otherwise specified, the norm of a vector $\bx$ is 
the $\ell_2$ norm $\| \bx \| := \sqrt{\sum_i (\bx(i))^2}$ and the norm of a matrix 
$A$ is the spectral norm $\|A\| := \sup_{\bx: \| \bx \|=1} \| A\bx \|$.
For a diagonal matrix, this is the largest diagonal entry in absolute value.

%% file: trunk/regular-full.tex
\section{Strong reconstruction for regular graphs} \label{sec:regular}
If $G$ is $d$-regular then $P = (1/d) A$ is a real symmetric matrix and 
$P$ and $A$ have the same set of eigenvectors.
We denote by $\bv_1 = (1/\sqrt { 2n }) \bone, \bv_2, \dots, \bv_{2n}$ a 
basis of orthonormal eigenvectors, where each $\bv_i$ is the eigenvector associated to
eigenvalue $\lambda_i$. Then, we can write a vector $\bx$
as a linear combination $\bx = \sum_i \alpha_i \bv_i$ and we have:
\begin{equation} 
    P^t \bx = \sum_i \lambda^t_i \alpha_i \bv_i 
    = \frac{1}{2n} \left(\sum_i \bx(i) \right) \bone + \sum_{i=2}^{2n} \lambda_i^t \alpha_i \bv_i,
    \label{eq:decompreg} 
\end{equation}
which implies that
$\mathbf{x}^{(t)} = P^t \bx$ tends to $\alpha_1 \bv_1$ as $t$ tends to 
infinity, i.e., it converges to the vector that has the 
average of $\bx$ in every coordinate.

We next show that, if the regular graph is ``well'' 
clustered, then the \avg protocol produces a strong reconstruction of the
two clusters w.h.p.

\begin{definition}[Clustered Regular Graph]\label{def:clusteredregular}
A $(2n,d,b)$-clustered regular graph 
$G = ((V_1,V_2),E)$ is a graph over vertex set $V_1 \cup V_2$, with 
$|V_1|=|V_2|=n$ and such that: (i) Every node  has degree $d$ and 
(ii) Every node in cluster $V_1$ has $b$ neighbors in cluster $V_2$ and 
every node in $V_2$ has $b$ neighbors in $V_1$.
\end{definition}

 
We know that $\bone$ is an eigenvector of $P$ with eigenvalue $1$, and it is 
easy to see that the partition indicator vector $\bchi$ is an eigenvector of 
$P$ with eigenvalue $1 - 2b/d$ (see Observation~\ref{apx-obs:secondeig} 
in Appendix~\ref{apx:la}).
We first  show  that, if $1- 2b/d$ happens to be the second 
eigenvalue, after $t$ rounds of the \avg 
dynamics, the configuration $\bx^{(t)}$  is close to a linear 
combination of $\bone$ and $\bchi$. Formally,
if $\lambda < 1 - 2b/d$ we prove (see Lemma~\ref{lem:metastable} in 
Appendix~\ref{sec:apx-wc}) that there are reals $\alpha_1,\alpha_2$ such that 
for every $t$

\begin{equation} \label{eq:metaregular}
\mathbf{x}^{(t)} = \alpha_1 \bone + \alpha_2 \lambda_2^t \bchi + \be^{(t)} \ \
\mbox{ where  } \ \left\| \be^{(t)} \right\|_\infty \leq \lambda^t \sqrt {2n} \, .
\end{equation}

Informally speaking, the equation above naturally ``suggested''  the 
choice of the coloring rule in the \avg protocol, once we considered 
the difference of two consecutive values of any node $u$, i.e.,

\begin{equation} \label{eq:diffregular}
\bx^{(t-1)}(u) - \bx^{(t)}(u) 
= \alpha_2 \lambda_2^{t-1} (1 - \lambda_2) \bchi(u) 
+ \be^{(t-1)}(u) - \be^{(t)}(u) \, .
 \end{equation}
Intuitively, if $\lambda$ is sufficiently small, we can exploit 
the bound on $\left\| \be^{(t)} \right\|_\infty$ in \eqref{eq:metaregular}
to show that, after a short initial phase,  the sign of $\bx^{(t-1)}(u) - \bx^{(t)}(u)$ is essentially determined   by 
$\bchi(u)$, thus by the community $u$ belongs to, w.h.p. 
The following theorem and its proof provide formal statements of the above fact.

\begin{theorem}[Strong reconstruction]\label{thm:reg_stop}
Let $G = ((V_1, V_2),E)$ be a connected $(2n,d,b)$-clustered regular graph with 
$1 - 2b/d > (1+\delta) \lambda$ for an arbitrarily-small constant $\delta > 0$.
Then the \avg protocol produces a \strong within $\bigO(\log n)$ rounds, w.h.p.
\end{theorem}
\skproof
From~\eqref{eq:diffregular}, we have that 
$\texttt{sgn}\left( \bx^{(t-1)}(u) - \bx^{(t)}(u) \right) 
= \texttt{sgn}\left(\alpha_2 \bchi(u) \right)$ whenever
\begin{equation}\label{eq:clustcond}
\left| \alpha_2 \lambda_2^{t-1} (1 - \lambda_2) \right| > 
\left| \be^{(t-1)}(u) - \be^{(t)}(u) \right|
\end{equation}
From~\eqref{eq:metaregular} we have that 
$\left| \be^{(t)}(u) \right| \leqslant \lambda^t \sqrt{2n}$, 
thus~\eqref{eq:clustcond} is satisfied for all $t$ such that
\[
t - 1 \geqslant 
\log \left(\frac{2 \sqrt{2n}}{|\alpha_2| (1 - \lambda_2)}
\right) \cdot \frac{1}{\log\left(\lambda_2 / \lambda \right)}\,.
\]

The second key-step of the proof relies on the randomness of the initial 
vector. Indeed, since $\bx$ is a vector of independent and uniformly 
distributed random variables in $\{-1,1\}$, the absolute difference 
between the two partial averages in the two communities, i.e. $|\alpha_2|$, is 
``sufficiently'' large, w.h.p. 
More precisely, observe that both $\langle \bx, \bchi 
\rangle$ and $\langle \bx , \bone \rangle$ have the distribution of 
a sum of $2n$ Rademacher random variables. Such a sum takes the 
value $2k-2n$ with probability $\frac{1}{2^n} \binom{2n}{k}$, and 
so every possible value has probability at most $\frac{1}{2^n}\binom{2n}{n} 
\approx \frac{1}{\sqrt {2\pi n}}$. Consequently, if $R$ 
is the sum of $2n$ Rademacher random variables, we have
$ \pr{}{ |R| \leq \delta \sqrt {2n}} \leq \bigO(\delta)$. 
This    implies that
$|\alpha_2| = \frac{1}{2n} \langle \bchi , \bx \rangle \geqslant n^{-\gamma}$, 
for some positive constant $\gamma$, w.h.p. (see Lemma~\ref{rad.sum}).
The theorem thus follows from the above bound on $|\alpha_2|$ and from the 
hypothesis $\lambda_2 \geqslant  (1+\delta)\lambda$. 
\qed

\medskip\noindent
\emph{Remark}. Graphs to which Theorem~\ref{thm:reg_stop} apply are those 
consisting of two regular expanders connected by a regular sparse cut.
Indeed, let $G = \left((V_1, V_2), E \right)$ be a $(2n, d,
b)$-clustered regular graph, and let $\lambda_{A}=\max\left\{ \lambda_2 (A_1),
\lambda_2(A_2)\right\}$ and $\lambda_B=\lambda_2(B)$, where 
$A_1$, $A_2$ and $B$ are the adjacency
matrices of the subgraphs induced by $V_1$, $V_2$ and the cut between $V_1$ and
$V_2$, respectively. 
Since $\lambda= \frac ad \lambda_{A}+ \frac bd \lambda_B$, if  
$a-b >(1+\epsilon)(a \lambda_A + b\lambda_B)$, $G$
satisfies the hypothesis of Theorem~\ref{thm:reg_stop}.

\paragraph{Regular stochastic block model.}
We can use Theorem~\ref{thm:reg_stop} to  prove  that
the \avg protocol achieves strong reconstruction in the regular stochastic 
block model. In the case of two communities, a graph on $2n$ vertices is 
obtained as follows:
Given two parameters $a(n)$ and $b(n)$ (\emph{internal} and \emph{external} 
degrees, respectively), partition vertices into two 
equal-sized subsets $V_1$ and $V_2$ and then sample a random 
$a(n)$-regular graph over each of $V_1$ and $V_2$ and a random 
$b(n)$-regular graph between $V_1$ and $V_2$. 
This model can be instantiated in different ways depending on how one samples 
the random regular graphs (for example, via the uniform distribution over 
regular graphs, or by taking the disjoint union of random matchings) \cite{mossel_reconstruction_2014,brito_recovery_2015}.

If $G$ is a graph sampled from the regular stochastic block model with 
internal and external degrees $a$ and $b$ respectively, then it is a
$(2n,d,b)$-clustered graph with largest eigenvalue of the 
transition matrix $1$ and corresponding eigenvector $\bone$,   
while $\bchi$ is also an eigenvector, with eigenvalue 
$1-2b/d$, where $d:= a+b$.
Furthermore, we can derive  the following  upper bound on the maximal absolute 
value achieved by   the  other   $2n-2$  eigenvalues corresponding to 
eigenvectors orthogonal to $\bone$ and $\bchi$:
\begin{equation} \label{eq:regsbm}
\lambda \leqslant \frac{2}{a+b} (\sqrt {a+b-1} + o_n(1))
\end{equation}

\noindent
This bound can be proved using some general result of Friedman and 
Kohler~\cite{FK14} on {\em random degree $k$ lifts} of a graph. 
(see Lemma \ref{ov:lm:regulsbm} in Appendix~\ref{sec:apx-random}).   
Since $\lambda_2=\frac{a-b}{a+b}$, using \eqref{eq:regsbm} 
in  Theorem \ref{thm:reg_stop}, we  get a strong 
reconstruction for the regular stochastic block model:
 
\begin{corollary}\label{cor.regular}
Let $G$ be a random graph sampled from the regular stochastic block model with
$a - b > 2(1 + \eta) \sqrt{a+b}$ for any constant $\eta>0$, then the 
\avg protocol produces a \strong in $\bigO(\log n)$ rounds,  w.h.p.
\end{corollary}

%% file: trunk/nonregular-full.tex
\section{Weak reconstruction for non-regular graphs} \label{sec:non-regular-det}
The results of Section \ref{sec:regular} rely on very clear spectral 
properties of regular, clustered graphs, 
immediately reflecting their underlying topological structure. 
Intuition suggests that these properties should be 
approximately preserved if we suitably relax the notion of regularity.
With this simple intuition in mind, we generalize our approach for regular graphs to a
large class of non-regular clustered graphs. 
\begin{definition}[Clustered $\gamma$-regular graphs]\label{def:epsreggraph} 
A $(2n,d,b,\gamma)$-clustered graph $G = ((V_1,V_2),E)$ is a graph over 
vertex set $V_1 \cup V_2$, where $|V_1|=|V_2|=n$ such that:
i) Every node has degree $d \pm \gamma d$, and 
ii) Every node in $V_1$ has $b \pm \gamma d$ neighbors in $V_2$ and every
node in $V_2$ has $b\pm \gamma d$ neighbors in $V_1$.
\end{definition}

If $G$ is not regular then matrix $P = D^{-1} A$ is not symmetric in general, 
however it is possible to relate its eigenvalues and eigenvectors to those of 
a symmetric matrix as follows.
Denote the {\em normalized adjacency matrix} of $G$ as 
$N:= D^{-1/2} A D^{-1/2} = D^{1/2} P D^{-1/2}$. Notice that 
$N$ is symmetric, $P$ and $N$ have the same eigenvalues $\lambda_1,\ldots, \lambda_{2n}$, and $\bx$ is an 
eigenvector of $P$ if and only if $D^{1/2}\bx$ is an eigenvector of $N$
(if $G$ is regular then $P$ and $N$ are the same matrix). 
Let $\bw_1,\ldots,\bw_{2n}$ be a basis of orthonormal eigenvectors of $N$, 
with $\bw_i$ the eigenvector associated to eigenvalue $\lambda_i$, for every $i$.
We have that $\bw_1 = \frac {D^{1/2} \bone}{\| D^{1/2} \bone \|}$.
If we set $\bv_i := D^{-1/2} \bw_i$, we obtain a set of eigenvectors for $P$ 
and we can write $\bx = \sum_i \alpha_i \bv_i$ as a linear combination of 
them. Then, the averaging process can again be described as 
\[ 
P^t \bx = \sum_i \lambda_i^t \alpha_i \bv_i = \alpha_1 \bv_1 
+ \sum_{i=2}^{2n} \lambda_i^t \alpha_i \bv_i.
\]
So, if $G$ is connected and not bipartite, the \avg dynamics
converges to $\alpha_1 \bv_1$. In general, it is easy to see that $\alpha_i = 
\bw_i^T D^{1/2} \bx $ (see the first lines in the proof of Lemma 
\ref{lem:clustered}) and $ \alpha_1 \bv_1$ is the vector 
\[ 
(\bw_1^T D^{1/2} \bx) \cdot D^{-1/2} \bw_1 
= \frac{\bone^T D \bx}{\|D^{1/2} \bone \|^2} \bone 
= \frac{ \sum_i d_i \bx(i)}{\sum_i d_i} \cdot \bone \,.
\]

\noindent
As in the regular case, if the transition matrix $P$ of a clustered 
$\gamma$-regular graph has $\lambda_2$ close to $1$ and $|\lambda_3|, 
\ldots, | \lambda_{2n}|$ small, the \avg dynamics has a long phase in which 
$\mathbf{x}^{(t)} = P^t \mathbf{x}$ is close to $\alpha_1 \bone + \alpha_2 
\mathbf{v}_2$. 

However, providing an argument similar to the regular case is 
considerably harder, since the 
partition indicator vector $\bchi$ is no longer an eigenvector of $P$. In order to 
fix this issue, we generalize \eqref{eq:metaregular}, proving in  
Lemma~\ref{lem:clustered} that $\bx^{(t)}$ is still close to a linear  
combination of $\bone$ and $\bchi$. 
We set $\nu = 1 - \frac{2b}{d}$, since this value occurs frequently in 
this section.


\begin{lemma}\label{lem:clustered}
Let $G$ be a connected $(2n,d,b,\gamma)$-clustered graph with $\gamma
\leqslant 1/10$, and assume the \avg dynamics is run on $G$ with 
initial vector $\bx$.
If $\lambda < \nu $ we have:
\[
    \mathbf{x}^{(t)} = \alpha_1 \bone + \alpha_2 \lambda_2^t  \bchi 
    + \alpha_2 \lambda_2^t  \bz + \be^{(t)}\,,
\]
for some vectors $\bz$ and $\be^{(t)}$ with 
$\|\bz\| \leq \frac{88\, \gamma}{\nu - \lambda_3}\sqrt{2 n}$ and
$\|\be^{(t)}\| \leq 4 \lambda^t \| \bx \|$.
Coefficients $\alpha_1$ and $\alpha_2$ are 
$\alpha_1 = \frac{\bone^{\intercal} D\bx}{\| D^{\frac 12} \bone \|^2}$
and
$\alpha_2 = \frac{\bw_2^{\intercal}D^{\frac 12}\bx}{\bw_2^{\intercal}D^{\frac 12} \bchi}$.
\end{lemma}

\skproof
We prove the following two key-facts:  (i) 
 the second eigenvalue of the 
transition matrix of $G$ is not much smaller than $1 - 2b/d$,  and (ii) $D^{1/2} \bchi$ 
is close, in norm, to its projection on the second eigenvector 
of the normalized adjacency matrix $N$. Namely,  in Lemma~\ref{lemma:sqrtdchiapprox} 
  we prove that
if $\lambda_3 < \nu $ then
\begin{equation}\label{eq:coleeigenvec}
    \lambda_2 \geq \nu - 10\gamma 
    \qquad 
    \mbox{ and }
    \qquad
    \left\| D^{1/2} \bchi - \beta_2 \bw_2 \right\| 
    \leqslant \frac{44\, \gamma}{\nu - \lambda_3}\sqrt{2 nd},
\quad \mbox{where } \beta_2 = \bchi^\intercal D^{1/2} \bw_2\,.
\end{equation}

Now, we can use the above bounds to analyze $\bx^{(t)} = P^t \bx$.
To begin, note that $N=D^{-1/2}AD^{-1/2}$ and $P=D^{-1}A$ imply that $P = D^{-1/2} N D^{1/2}$ 
and $P^t = D^{-1/2} N^tD^{1/2}$. Thus, for any vector $\bx$, if we write 
$D^{1/2}\bx$ as a linear combination of an orthonormal basis of $N$, 
$D^{1/2} \bx = \sum_{i = 1}^{2n} a_i \bw_i$, we get 
\[
    P^t \bx = D^{-1/2} N^t D^{1/2} \bx 
    = D^{-1/2} \sum_{i = 1}^{2n} a_i \lambda_i^t \bw_i
    = \sum_{i = 1}^{2n} a_i \lambda_i^t D^{-1/2} \bw_i.
\]
We next estimate the first term, the second term, and the 
sum of the remaining terms:

\noindent
- We have $\bw_1 = \frac{D^{1/2} \bone}{\|D^{1/2} \bone\|}$, so the first term 
can be written as $\alpha_1 \bone$ with
$\alpha_1 = \frac{a_1}{\left\| D^{1/2} \bone \right\|}  
= \frac{\bw_1^\intercal D^{1/2} \bx}{\left\| D^{1/2} \bone \right\|}
= \frac{\bone^\intercal D \bx}{\left\| D^{1/2} \bone \right\|^2}$.

\noindent
- If we write $D^{1/2} \bchi = \beta_2 \bw_2 + \by$, with 
$\beta_2 = \bw_2^\intercal D^{1/2} \bchi$, \eqref{eq:coleeigenvec} 
implies that 
   $ \| \by \| \leqslant \frac{44\, \gamma}{\nu - \lambda_3}\sqrt{2 nd}$.
Hence the second term can be written as
\[
    a_2 \lambda_2^t D^{-1/2} \bw_2 
    = a_2 \lambda_2^t D^{-1/2} \left( \frac{D^{1/2} \bchi - \by}{\beta_2} \right)
    = \frac{a_2}{\beta_2} \lambda_2^t \bchi - \frac{a_2}{\beta_2} \lambda_2^t \bz
    = \alpha_2 \lambda_2^t \bchi - \alpha_2 \lambda_2^t \bz,
\]
where 
\[
    \| \bz \| = \left\| D^{-1/2} \by \right\| 
    \leqslant \left\| D^{-1/2} \right\| \, \| \by \|
    \leqslant \frac{2}{\sqrt{d}} \cdot \frac{44\, \gamma}{\nu - \lambda_3}\sqrt{2 nd}
= \frac{88\, \gamma}{\nu - \lambda_3}\sqrt{2 n},  \]
and
\[ \alpha_2 = a_2/\beta_2 = \frac{\bw_2^\intercal D^{1/2} \bx }{ \bw_2 D^{1/2} \bchi}. \]

\noindent
- As for all other terms, observe that
\begin{multline}
   \|\be^{(t)}\|^2 
   = \left\| \sum_{i = 3}^{2n} a_i \lambda_i^t D^{-1/2}\bw_i \right\|^2
         \leqslant  \left\| D^{-1/2} \right\|^2 
        \, \left\| \sum_{i=3}^{2n} a_i \lambda_i^t \bw_i \right\|^2\\
        = \left\| D^{-1/2} \right\|^2 \, \sum_{i = 3}^{2n} a_i^2 \lambda_i^{2t}
   \leqslant \left\| D^{-1/2} \right\|^2 \, 
    \lambda^{2t} \sum_{i = 3}^{2n} a_i^2  
    \leqslant \left\| D^{-1/2} \right\|^2 
    \, \lambda^{2t} \left\| D^{1/2} \bx \right\|^2 \\
     \leqslant  \left\| D^{-1/2} \right\|^2 \, \left\| D^{1/2} \right\|^2
    \, \lambda^{2t} \| \bx\|^2
    \leqslant 16 \lambda^{2t} \| \bx \|^2.
\end{multline}
\qed

\smallskip\noindent
The above lemma allows us to generalize our approach to achieve 
efficient, weak reconstruction in non-regular clustered 
graphs. The full proof of the following theorem is given in  
appendix \ref{apx-theoremweak}.

\begin{theorem}[Weak reconstruction] 
    \label{thm:irreg_stop}
    Let $G$ be a connected $(2n, d, b, \gamma)$-clustered graph with 
    $\gamma \leqslant c(\nu-\lambda_3)$ for a suitable constant $c>0$. 
    If $\lambda < \nu $ and $\lambda_2 \geqslant (1 + \delta) \lambda$ 
    for an arbitrarily-small positive constant $\delta$, the \avg protocol
    produces an $\bigO ( {\gamma^2 }/(\nu - \lambda_3)^2 )$-weak reconstruction 
    within $\bigO(\log n)$ rounds w.h.p.\footnote{Consistently, 
    Theorem~\ref{thm:reg_stop} is a special case of this one when
    $\gamma = 0$.} 
\end{theorem}
\skproof
    Lemma~\ref{lem:clustered} implies that for every node $u$ at any 
    round $t$ we have
    \[
        \bx^{(t-1)}(u) - \bx^{(t)}(u) 
        = \alpha_2 \lambda_2^{t-1} (1 - \lambda_2) \left(\bchi(u) + \bz(u)\right)
        + \be^{(t-1)}(u) - \be^{(t)}(u)
    \]
    Hence, for every node $u$ such that $|\bz(u)| < 1/2$,\footnote{The 
    value $1/2$ is chosen here only for readability sake, any     
    constant smaller than $1$ will do.} we have
    $\texttt{sgn}\left( \bx^{(t-1)}(u) - \bx^{(t)}(u) \right) 
    = \texttt{sgn}\left(\alpha_2 \bchi(u) \right)$ whenever
    \begin{equation}\label{eq:clustcondnonreg}
        \left| \frac{1}{2} \alpha_2 \lambda_2^{t-1} (1 - \lambda_2) \right| > 
        \left| \be^{(t-1)}(u) - \be^{(t)}(u) \right|.
    \end{equation}
    From Lemma~\ref{lem:clustered} we have 
    $\left| \be^{(t)}(u) \right| \leqslant 4 \lambda^t \sqrt{2n}$, 
    thus~\eqref{eq:clustcondnonreg} is satisfied for any  $t$ such that
    \begin{equation}
        t - 1 \geqslant 
        \log \left(\frac{16 \sqrt{2n}}{|\alpha_2| (1 - \lambda_2)}
        \right) \cdot \frac{1}{\log\left(\lambda_2 / \lambda \right)} \, .
        \label{eq:time}
    \end{equation}
    The right-hand side of the above formula is $\mathcal{O}(\log n)$ 
    w.h.p., because of the following three points:
 i) $\lambda_2 \geqslant (1+\delta)\lambda$ by hypothesis;
         ii) $1 - \lambda_2 \geqslant 1/(2n^4)$ from Cheeger's 
         inequality (see e.g. \cite{chung96}) and the fact
            that the graph is connected; iii)
     using similar (although harder - see Lemma~\ref{lemma:apxrdm}) 
     arguments as in  the proof of Theorem \ref{thm:reg_stop}, we can prove that
     Rademacher initialization of $\bx$ w.h.p. implies  $|\alpha_2| \geqslant n^{-c}$ for some large enough
            positive constant $c$.
 Finally, from Lemma~\ref{lem:clustered} we have 
       $ \|\bz \| \leqslant \frac{88\, \gamma}{\nu - \lambda_3}\sqrt{2 n}$.
    Thus, the number of nodes $u$ with $\bz(u) \geqslant 1/2$ is $\bigO(n {\gamma^2 }/(\nu - \lambda_3)^2)$.
\qed

\medskip\noindent
Roughly speaking, the above theorem states that the quality of block reconstruction 
depends on the regularity of the graph (through parameter $\gamma$) and  
conductance within each community (here represented by the difference
$|\nu - \lambda_3|$). Interestingly enough, as long as $|\nu - 
\lambda_3| = \Theta(1)$, the protocol achieves $\bigO(\gamma^2)$-weak 
reconstruction on $(2n, d, b, \gamma)$-clustered graphs.

\paragraph*{Stochastic block model.}
Below we prove that the stochastic block model $\planted$ satisfies
the hypotheses of Theorem~\ref{thm:irreg_stop}, w.h.p., and, thus,  the \avg protocol 
efficiently produces a good reconstruction. In what follows, we will often use 
the following parameters of the model:
expected internal degree $a = p n$, expected external degree $b = qn$, and 
$d = a + b$. 

\begin{lemma}\label{lemma:irreg} 
Let $G\sim \mathcal{G}_{2n,p,q}$. If $a-b > \sqrt{(a+b) \log n}$ then a 
positive constant $\delta$ exists such that the following  hold w.h.p.: 
i) $G$ is $(2n,d,b, 6\sqrt {\log n / d} )$-clustered \  \emph{and} \ 
ii) $\lambda \leq \min \left\{\lambda_2/(1+\delta) \,, \, 
24\sqrt{ (\log n)/d } \right\}$.
\end{lemma}
\skproof
Claim~(i) follows (with probability $1-n^{-1}$) from an easy 
application of the Chernoff bound.
As for Claim~(ii), since $G$ is not regular and random, we derive spectral 
properties on its adjacency matrix $A$ by considering a ``more-tractable'' matrix, namely 
the expected matrix
\begin{small}
\[ 
    B:= \E{A} = \left( 
        \begin{array}{cc}
        pJ, & qJ\\
        qJ, & pJ
        \end{array}
    \right)
\]
\end{small}

\vspace{-3mm}\noindent
where $B_{i,j}$ is the probability that the edge $(i,j)$ exists in a random graph 
$G \sim  \mathcal{G}_{2n,p,q}$.  In Lemma~\ref{lemma:irregone} we will prove that such a 
$G$ is  likely to have an adjacency
matrix $A$ close to $B$ in spectral norm. Then, in Lemma~\ref{lemma:irregtwo} 
we will show that every clustered graph whose adjacency matrix is close to $B$
has the properties required in the analysis of the \avg dynamics, thus getting
Claim~(ii).
\qed

\medskip\noindent
By combining Lemma~\ref{lemma:irreg} and Theorem~\ref{thm:irreg_stop}, we  
achieve weak reconstruction for the stochastic block model.
\begin{corollary} \label{cor.irreg}
Let $G \sim \mathcal{G}_{2n,p,q}$. If $a-b > 25 \sqrt{d \log n}$ and
$b=\Omega(\log n /n)$ then the \aveprot{} protocol produces an
$\bigO({{{d}\log n }/(a-b)^2})$-\weak in $\bigO(\log n)$ rounds w.h.p.
\label{cor:weakrec_gnpq}
\end{corollary}
\skproof
From Lemma~\ref{lemma:irreg} we get that w.h.p. $G$ is
$(2n,d,b,\gamma)$-clustered with $\gamma \leq 6\sqrt{{\log n}/{d}} $,
$|\lambda_i | \leq 4\gamma  $ for all $i=3,\ldots,2n$ and $\lambda_2 \geq
(1+\delta)\lambda_3 $ for some constant $\delta>0$. 
Given the hypotheses on $a$ and $b$, we 
also have that the graph is connected w.h.p. 
Moreover, since $d\nu = (a-b) > 25\sqrt{ d \log n }$, then 
\[
    \frac{\gamma}{\nu - \lambda_3} 
    = \frac{d\gamma}{d\nu - d\lambda_3} 
    \leq \frac{6 \sqrt{d\log n}}{(a-b) - 24 \sqrt{d\log n}} 
    = \bigO\left( \frac{\sqrt{d\log n}}{(a-b)} \right).
\]
Theorem~\ref{thm:irreg_stop} then guarantees that the \avg protocol finds an 
$\bigO\left({{d\log n }/(a-b)^2}\right)$-\weak w.h.p.
\qed

%% file: trunk/tightblock.tex
\subsection{Tight analysis for the stochastic block model} \label{sec:tightsbm} 
In Lemma~\ref{lemma:irreg} we have shown that, when $(a-b) >
\sqrt{(a+b)\log n}$, a graph sampled according to $\planted$ satisfies
the hypothesis of Theorem~\ref{thm:irreg_stop} w.h.p.: The simple \avg protocol
thus gets weak-reconstruction in $\bigO(\log n)$ rounds. As for the 
parameters' range of $\planted$, we know that the above result is still off by 
a factor $\sqrt{\log n}$ from the threshold $(a-b) > 2\sqrt{(a+b)}$
\cite{mossel_proof_2013,massoulie_community_2014,mossel_reconstruction_2014},
the latter being a necessary condition for any (centralized or not)
non-trivial  weak reconstruction. Essentially, the reason behind this gap is
that, while Theorem~\ref{thm:irreg_stop} holds for 
\emph{any} (i.e. ``worst-case'') $(2n,d,b,\gamma)$-clustered graph, in order 
to apply it to $\planted$ we need to choose parameters $a$ and $b$ in a way 
that $\gamma d$ bounds the variation of the degree of \emph{any} node w.r.t. 
the regular case w.h.p.

On the other hand, since the degrees in $\planted$ are distributed according
to a sum of Bernoulli random variables, the rare event that some degrees are 
much higher than the average does not affect too much the eigenvalues and
eigenvectors of the graph. 
Indeed, by adopting ad-hoc arguments for $\planted$, we prove
that the \avg protocol actually achieves an $\bigO(d / (a-b)^2)$-weak 
reconstruction w.h.p., provided
that $(a-b)^2 > \Cl[small]{optconst} (a+b) > 5\log n$, thus matching the 
weak-reconstruction threshold up to a constant factor for graphs of logarithmic degree. The main argument 
relies on the spectral properties of $\planted$ stated in the following 
lemma, whose complete proof is given in Appendix~\ref{sec:apx-random}. 

\begin{lemma}\label{lem:main}  
Let $G \sim \planted$. If $(a-b)^2 > \Cr{optconst}(a+b) > 5\log n$
and\footnotemark[\value{footnote}] $a+b<n^{\frac 13-\Cr{tight}}$ for some
positive constants $\Cr{optconst}$ and $\Cr{tight}$, then the following claims hold w.h.p.:
\begin{enumerate}
\item $\lambda_2 \geq 1- \nfrac{2b}{d} -
\Cl[small]{second_eigen_err}/\sqrt{d}$ for some constant
$\Cr{second_eigen_err}>0$,
\item $\lambda_2 \geq (1+\delta)\lambda$ for some constant $\delta>0$ 
(where as usual  $\lambda = \max \{ |\lambda_3|, \ldots ,  |\lambda_{2n}|\}$),
\item   $| \sqrt{2nd}( D^{-1/2} \bw_2 )(i) - \bchi(i) | \leq \frac 1 {100}$
for each   $i\in V \setminus S$,  for some subset $S$ with 
$|S| = \bigO(n d / (a-b)^2)$.
\end{enumerate}
\end{lemma}
\ideaproof
    The key-steps of the proof are   two     probability-concentration
    results. 
    In Lemma \ref{lem:N_vs_B}, we prove a tight bound on the deviation of the
    Laplacian $\mathcal{L}(A) = I-N$ of    $\planted$ from the Laplacian of the
    expected matrix $\mathcal{L}(B)=I-\frac{1}{d}B$. As one may expect from
    previous results on the Erd\H{o}s-R\'enyi model and from Le and Vershynin's
    recent concentration results for inhomogeneous Erd\H{o}s-R\'enyi graph (see
    \cite{le_concentration_2015}), we can prove that w.h.p. $\|\mathcal{L}(A) -
    \mathcal{L}(B)\| = \bigO(\sqrt{d})$, even when $d = \Theta(\log n)$. To
    derive the latter result, we leverage on the aforementioned Le and
    Vershynin's bound on the spectral norm of inhomogeneous Erd\H{o}s-R\'enyi
    graphs; in    $\planted$ this bound implies that if $d = \Omega(\log n)$ then
    w.h.p. $\|A-B\|=\bigO(\sqrt{d})$. Then, while Le and Vershynin replace the
    Laplacian matrix with regularized versions of it, we are able to bound
    $\|\mathcal{L}(A) - \mathcal{L}(B)\|$ directly by upper bounding it with
    $\|A-B\|$ and an additional factor $\|B-d^{-1}\,D^{1/2}BD^{1/2}\|$.
    We then bound from above the latter additional factor thanks to our second
    result: In Lemma~\ref{lem:sumdeg}, we prove that w.h.p. $\sum (\sqrt{d_i} -
    \sqrt{d})^2\leq 2n$ and $\sum (d_i - d)^2 \leq 2nd$. 
    We can then prove the first two claims of Lemma~\ref{lem:main} by 
    bounding the distance of the eigenvalues of $N$ from those of $d^{-1}\,B$
    via Lemma \ref{thm:eigen_vs_norm}. 
    As for the third claim of the lemma, we prove it by upper bounding the
    components of $D^{-1/2}\bw$ orthogonal to $\bchi$. In particular, we can
    limit the projection $\bw_{\bone}$ of $D^{-1/2}\bw$ on $\bone$ by
    using Lemma~\ref{lem:sumdeg}. Then, we can upper bound the projection
    $\bw_{\perp}$ of $D^{-1/2}\bw$ on the space orthogonal to both $\bchi$ and
    $\bone$ with Lemma~\ref{lem:N_vs_B}: We look at $N$ as a perturbed version
    of $B$ and apply the Davis-Kahan theorem. Finally, we conclude the proof
    observing that $\|\bw_2 - (2n)^{-1/2}\|\leq 2(\|\bw_{\bone}\|+\|\bw_{\perp}\|)$. 
\qed

\smallskip\noindent
Once we have Lemma~\ref{lem:main} we can prove the main theorem on    $\planted$ 
with the same argument used for Theorem~\ref{thm:irreg_stop} (the full proof 
is given in Appendix~\ref{sec:apx-random}).

\begin{theorem}\label{thm:tight_bounds}
Let $G \sim \mathcal{G}_{2n,p,q}$. If $(a-b)^2 > \Cr{optconst}(a+b) > 5 \log n$ 
and\footnote{\label{note1}It
should be possible to weaken the condition $d< n^{\frac 13-\Cr{tight}}$ via
some stronger concentration argument; see the proof of Lemma \ref{lem:sumdeg}
in Appendix~\ref{apx:sum_of_deg} for details.} $a+b<n^{\frac 13-\Cr{tight}}$
for some positive constants $\Cr{optconst}$ and $\Cr{tight}$, then the \avg
protocol produces an $\bigO(d/(a-b)^2)$-weak reconstruction within $\bigO(\log
n)$ rounds w.h.p. 
\end{theorem}

%% file: trunk/conclusions.tex
\section{Moving beyond two communities: An outlook} 
\label{sec:morecomm}
The \avg protocol can be naturally extended to address the case of 
more communities. One way to achieve this is by performing a suitable number 
of independent, parallel runs of the protocol. We next outline the analysis 
for a natural generalization of the regular block model. 
This allows us to easily present the main ideas and to provide an 
intuition of how and why the protocol works. 

Let $G= (V,E)$ be a $d$-regular graph in which $V$ is partitioned into $k$
equal-size communities $V_1,\ldots,V_k$, while every node in $V_i$ has exactly $a$
neighbors within $V_i$ and exactly $b$ neighbors in each $V_j$, for $j\neq i$.
Note that $d = a + (k-1)\cdot b$. It is easy to see that the transition matrix
$P$ of the random walk on $G$ has an eigenvalue $(a - b)/d$ with multiplicity
$k-1$. The eigenspace of $(a-b)/d$ consists of all stepwise vectors that are constant
within each community $V_i$ and whose entries sum to zero. If $
\max\{|\lambda_{2n}|, \lambda_{k+1}\} < (1-\epsilon) \cdot (a - b)/d$, $P$
has eigenvalues $\lambda_1 = 1$, $\lambda_2 = \dots = \lambda_k = (a - b)/d$,
with all other eigenvalues strictly smaller by a $(1-\epsilon)$ factor.

Let $T$ be a large enough threshold such that, for all $t \geq T$, 
$\lambda_2^t > n^2 \lambda_{k+1}^t$ and note that $T$ is in the order of 
$(1/\epsilon) \log n$. 
Let $\bx \in \R^V$ be a vector. We say that a vertex $v$ is of {\em negative 
type} with respect to $\bx$ if, for all $t>T$, the value $(P^t \bx)_v$ 
decreases with $t$. We say that a vertex $v$ is of {\em positive type} with 
respect to $\bx$ if, for all $t>T$, the value $(P^t \bx)_v$ increases with $t$.
Note that a vertex might have neither type, because $(P^t \bx)_v$ might not be 
strictly monotone in $t$ for all $t> T$.

In Appendix~\ref{sec:apx-more} we prove the following: If we pick $\ell$ random vectors 
$\bx^1,\ldots, \bx^\ell$, each in $\{ -1,1\}^V$ then, with high probability,  
i) every vertex is either of positive or negative type for each 
$\bx^i$;\footnote{I.e., for every $t > T$, $(P^t \bx)_v$ monotonically 
increases (or decreases) with $t$.} ii) furthermore, if 
we associate a ``signature'' to each vertex, namely, the sequence of $\ell$ 
types, then vertices within the same $V_i$ exhibit the same signature, 
while vertices in different $V_i,V_j$ have different signatures. These 
are the basic intuitions that allow us to prove the following theorem.

\begin{theorem}[More communities]
    \label{thm:more}
    Let $G = (V,E)$ be a $k$-clustered $d$-regular graph defined as above and
    assume that $\lambda = \max\{|\lambda_{2n}|, \lambda_{k+1}\} < (1 -
    \epsilon)\frac{a - b}{d}$, for a suitable constant $\epsilon > 0$. Then,
    for $\ell = \Theta(\log n)$, the \avg protocol with $\ell$ parallel runs
    produces a \strong within $\bigO(\log n)$ rounds, w.h.p.
\end{theorem}

%% file: trunk/apx-preli.tex
\section{Linear algebra toolkit}
\label{apx:la}

If $M \in  \R^{n \times n}$ is a real symmetric matrix, then it has $n$ real 
eigenvalues (counted with repetitions), $\lambda_1 \geq \lambda_2 \geq \cdots 
\geq \lambda_n$, and we can find a corresponding collection of orthonormal 
real eigenvectors $\bv_1,\ldots,\bv_n$ such that $M \bv_i = \lambda_i \bv_i$.

If $\bx \in \R^n$ is any vector, then we can write it as a linear combination
$\bx = \sum_i \alpha_i \bv_i$ of eigenvectors, where the coefficients of the
linear combination are $\alpha_i = \langle \bx, \bv_i \rangle$. In this 
notation, we can see that
\[
M \bx = \sum_i \lambda_i \alpha_i \bv_i,
\quad
\mbox{ and so }
\quad
M^t \bx = \sum_i \lambda_i^t \alpha_i \bv_i.
\]

\begin{lemma}[Cauchy-Schwarz inequality]\label{apx:csineq}
For any pair of vectors $\bx$ and $\by$
\[ 
\left| \langle \bx, \by \rangle \right| \leq \| \bx \| \cdot \| \by \|.
\]
\end{lemma}

\begin{observation}
For any matrix $A$ and any vector $\mathbf{x}$
\[ 
\| A\bx \| \leq \| A \| \cdot \|\bx \|,
\quad
\mbox{ and } 
\quad
\| A \cdot B \| \leq \| A \| \cdot \| B \|.
\]
\end{observation}

\begin{observation}\label{apx-obs:secondeig}
If $G$ is a $(2n,d,b)$-clustered regular graph with clusters $V_1$ and 
$V_2$ and $\bchi = \bone_{V_1} - \bone_{V_2}$ is the partition indicator 
vector, then $\bchi$ is an eigenvector of the transition matrix $P$ of $G$ 
with eigenvalue $1- 2b/d$.
\end{observation}

\begin{proof} 
Every node $i$ has $b$ neighbors $j$ on the opposite side of the partition, 
for which  $\bchi(j) = - \bchi(i)$, and $d-b$ neighbors $j$ on the same side, 
for which $\bchi(j) = \bchi(i)$, so 
\[
(P\bchi)_i 
= \frac 1d \left( (d-b) \bchi(i) - b \bchi(i) \right) 
= \left( 1 - \frac{2b}{d} \right) \bchi(i).
\]
\end{proof}

\begin{theorem}[Matrix Bernstein Inequality]\label{thm:matrixBer}
Let $X_1,\ldots,X_N$ be a sequence of independent $n \times n$ symmetric 
random matrices, such that $\E {X_i} = \bzero$ for every $i$ and such that 
$\| X_i \| \leq L$ with probability 1 for every $L$. Call 
$\sigma := \| \E{ \sum_i X_i^2} \|$. Then, for every $t$, we have
\[ 
\pr{} { \left\| \sum_i X_i \right\| \geq t } \leq 2n e^{\frac {-t^2}{2 \sigma + \frac 23 Lt } }.
\]
\end{theorem}

\begin{theorem}(Corollary 4.10 in \cite{stewart})
    \label{thm:eigen_vs_norm}
    Let $M_1$ and $M_2$ be two Hermitian matrices, let $\lambda_1 \geq \lambda_2
    \geq \cdots \geq \lambda_n$ be the eigenvalues of $M_1$ with multiplicities in
    non-increasing order, and let $\lambda'_1 \geq \lambda'_2 \geq \cdots \geq
    \lambda'_n$ be the eigenvalues of $M_2$ with multiplicities in non-increasing
    order. Then, for every $i$, 
    \[ 
        | \lambda_i - \lambda'_i | \leq \| M_1 - M_2\|.
    \]
\end{theorem}

\begin{theorem}[Davis and Kahan, 1970]
    \label{thm:daviskahan}
    Let $M_1$ and $M_2$ be two symmetric real matrices, let $\bx$ be a unit 
    length eigenvector of $M_1$ of eigenvalue $t$, and let $\bx_{p}$ be the 
    projection of $\bx$ on the eigenspace of the eigenvectors of $M_2$ 
    corresponding to eigenvalues $\leq t-\delta$. Then
    \[ 
        \| \bx_p \| \leq \frac 2 {\delta \pi} \| M_1 - M_2\|.
    \]
\end{theorem}

\section{Length of the projection of $\bx$} 

For the analysis of the \avg 
dynamics on both regular and non-regular graphs, it is important to 
understand the distribution of the projection of $\bx$ on $\bone$ 
and $\bchi$, that is (up to scaling) the distribution of the inner 
products $\langle \bx , \bone\rangle$ and $\langle \bx, \bchi 
\rangle$.  In particular we are going to use the 
following bound.

\begin{lemma} \label{rad.sum} 
If we pick $\bx$ uniformly at random  in $\{-1,1\}^{2n}$ then, for any     $\delta >0$ and any fixed 
      vector $\bw \in \{-1,1\}^{2n}$ with $\pm 1$ entries, it holds 
    \begin{equation} 
        \pr{}{ \big|  \langle (1/ \sqrt{2n})\, \bw  ,\, \bx \rangle \big|   \leq \delta } \leq \bigO(\delta).
    \end{equation}
\end{lemma}

\begin{proof}
Since $\bx$ is a vector of independent and uniformly 
distributed random variables in $\{-1,1\}$, both $\langle \bx, \bchi 
\rangle$ and $\langle \bx , \bone \rangle$ have the distribution of 
a sum of $2n$ Rademacher random variables. Such a sum takes the 
value $2k-2n$ with probability $\frac{1}{2^n} \binom{2n}{k}$, and 
so every possible value has probability at most $\frac{1}{2^n}\binom{2n}{n} 
\approx \frac{1}{\sqrt {2\pi n}}$. Consequently, if $R$ 
is the sum of $2n$ Rademacher random variables, we have
$\Prob{}{ |R| \leq \delta \sqrt {2n}} \leq \bigO(\delta)$. 
\end{proof}

\noindent
Although it is possible to argue that  a Rademacher vector has $\Omega (1)$ probability of having inner product $\Omega( \| \bw\|)$ with every vector $\bw$, such a statement does not hold w.h.p.
We do have, however, estimates of the inner product of a   vector $\bw$ with a Rademacher vector $\bx$ provided that $\bw$ is close to a vector in $\{-1,1\}^{2n}$.

\begin{lemma}
    \label{lemma:apxrdm} 
    Let $k$ be a positive integer. For every $nk$-dimensional vector $\bw$ such that
    $|\left\{ i \,|\,\, |\bw(i)| \geq c \right\}| \geq n$ for some positive
    constant $c$, if we pick $\bx$ uniformly at random in $\{-1,1\}^{kn}$, then
    \begin{equation} 
        \pr{}{ \big|  \langle(1/ \sqrt{kn})\,\bw,\, \bx \rangle \big|  \leq \delta} \leq
            \bigO(k \delta) + \bigO\left( \frac 1 {\sqrt {n}} \right).
    \end{equation}
\end{lemma}
\begin{proof}
     Let $S \subset \{ 1,\ldots,kn\}$ be the set of  coordinates $i$ of $\bw$
     such that $| \bw(i) | \geq c$. By hypothesis, we have
     $|S| \geq n$. Let $T:= \{ 1,\ldots, kn\} - S$. Now, for every
     assignment ${\bf a}\in \{-1,1\}^{kn}$, we will show that 
    \[  
        \pr{}{ | \langle \bw,\bx \rangle |   \leq \delta \sqrt {kn} \ | \  \ \forall i\in T , \,
        \bx(i) = {\bf a}(i) } \leq \bigO(\delta),
    \]
    and then the lemma will follow. Call $t := \sum_{i\in T} a_iz_i$. We need to show
    \[  
        \pr{}{ | \sum_{i \in S} \bx(i) \bw(i) + t |  \leq \delta \sqrt {kn}  } \leq \bigO(\delta). 
    \]
    From the Berry-Esseen theorem, 
    \[   
        \pr{}{  | \sum_{i \in S} \bx(i) \bw(i) + t |  \leq \delta \sqrt {kn}  } \leq
        \pr{}{  | g+ t |  \leq \delta \sqrt {kn} } + \bigO \left( \frac 1 {\sqrt {n}}
        \right), 
    \]
    where $g$ is a Gaussian random variable of mean 0 and variance $\sigma^2 =
    \sum_{i\in S}  (\bw(i))^2 \geq c^2 |S| \geq c^2\,n$, so
    \[  
        \pr{}{  | g+ t |  \leq \delta \sqrt {kn} }  = \frac 1 { \sqrt{2\sigma^2\pi}}
        \int_{-t-\delta\sqrt {kn}}^{-t+\delta \sqrt {kn}} e^{-\frac{s^2}{2\sigma^2}} ds \leq
        \frac{2\delta \sqrt {kn}}{\sqrt{ 2\pi c^2\,n}} = \frac{\sqrt{2k} \delta }{\sqrt{\pi}c},
    \]
    where we used the fact that $e^{-s^2/2} \leq 1$ for all $s$.
\end{proof}

%% file: trunk/apx-wc.tex
\section{Clustered Graphs}\label{sec:apx-wc}

\begin{lemma}\label{lem:metastable} 
Assume we run the \avg dynamics in a $(2n, d, b)$-clustered regular graph 
$G$ (see Definition~\ref{def:clusteredregular}) with any initial vector $\mathbf{x} \in \{ -1,1\}^{2n}$.
If $\lambda < 1 - 2b/d$ then there are reals $\alpha_1,\alpha_2$ such that 
at every round $t$ we have 
\[
\mathbf{x}^{(t)} = \alpha_1 \bone + \alpha_2 \lambda_2^t \bchi + \be^{(t)} \ \ 
 \mbox{ where } \  \left\| \be^{(t)} \right\|_\infty \leq \lambda^t \sqrt {2n} \, . \]
\end{lemma}

\begin{proof}
Since $\bx^{(t)} = P^t \bx$ we can write
\[ 
P^t \bx = \sum_i \lambda_i^t \langle \bx, \bv_i \rangle\bv_i, 
\]
where $1 = \lambda_1 > \lambda_2 = 1 - 2b/d > \lambda_3 \geq \cdots \geq \lambda_{2n}$ 
are the eigenvalues of $P$ and $\bv_1 = \frac 1 {\sqrt {2n}} \bone$, 
$\bv_2 = \frac 1{ \sqrt {2n}} \bchi$, $\bv_3$, \ldots, $\bv_{2n}$ are a corresponding 
sequence of orthonormal eigenvectors. 
Hence, 
\begin{align} 
\mathbf{x}^{(t)} & = \frac 1 {2n} \langle \bx, \bone \rangle \cdot \bone 
+ \lambda_2^t \frac 1{2n} \langle \bx, \bchi \rangle \cdot \bchi 
+ \sum_{i=3}^{2n} \lambda_i^t \alpha_i \bv_i \\
& = \alpha_1 \bone + \alpha_2 \lambda_2^t \cdot \bchi 
+ \sum_{i=3}^{2n} \lambda_i^t \alpha_i \bv_i,
\end{align} 
where we set $\alpha_1 = \frac{1}{2n} \langle \bone,\bx \rangle$ and 
$\alpha_2 = \frac{1}{2n} \langle \bchi , \bx \rangle$. 
We bound the $\ell_\infty$ norm of the last term as
\[
    \left\| \sum_{i=3}^{2n} \lambda_i^t \alpha_i \bv_i \right\|_\infty 
    \leq \left\| \sum_{i=3}^{2n} \lambda_i^t \alpha_i \bv_i  \right\|_2 
    = \sqrt{\sum_{i=3}^{2n} \lambda_i^{2t} \alpha_i^2} \leq \lambda^t \sqrt{\sum_{i=1}^{2n} \alpha_i^2} 
    = \lambda^t \| \bx\| = \lambda^t \sqrt {2n}.
\]
\end{proof}

\begin{lemma}\label{lemma:sqrtdchiapprox}
Let $G$ be a connected $(2n, d, b, \gamma)$-clustered graph (see 
Definition~\ref{def:epsreggraph}) with $\gamma \leqslant 1/10$. If 
$\lambda_3 < \nu$ 
then
\[
\lambda_2 \geq \nu - 10\gamma 
\qquad 
\mbox{ and }
\qquad
\left\| D^{1/2} \bchi - \beta_2 \bw_2 \right\| 
\leqslant \frac{44\, \gamma}{\nu - \lambda_3}\sqrt{2 nd},
\]
where $\beta_2 = \bchi^\intercal D^{1/2} \bw_2$.
\end{lemma}

\begin{proof}
For every node $v$, let us name $a_v$ and $b_v$ the numbers of neighbors of $v$
in its own cluster and in the other cluster, respectively, and 
$d_v = a_v + b_v$ its degree. 
Since from the definition of $(2n, d, b, \gamma)$-clustered graph it holds
that $(1-\gamma) d \leqslant d_v \leqslant (1 + \gamma) d$ and
$b - \gamma d \leqslant b_v \leqslant b + \gamma d$, it is easy to 
check that
\[
| a_v - b_v - \nu d_v | \leqslant 4 d \, \gamma
\]
for any node $v$. Hence, 
\begin{align}
    \left\| A \bchi - \nu D \bchi \right\|^2 
    & = \sum_{v \in [2n]} \left(\sum_{w \in \text{Neigh}(v)} \bchi(w) - 
    \nu d_v \bchi(v) \right)^2 \\
    & = \sum_{v \in [2n]} \left(a_v \bchi(v) - b_v \bchi(v) 
    - \nu d_v \bchi(v) \right)^2 \\
    & = \sum_{v \in [2n]} \left( a_v - b_v - \nu d_v \right)^2
    \leqslant 32 n d^2 \gamma^2. 
\end{align}
Thus,
\begin{align}\label{eq:sqrtdchibound}
    \left\| N D^{1/2} \bchi - \nu D^{1/2} \bchi \right\|
    & = \left\| D^{-1/2} A \bchi - \nu D^{1/2} \bchi \right\|
    = \left\| D^{-1/2} \left( A \bchi - \nu D \bchi \right) \right\| 
    \nonumber \\ 
    & \leqslant \left\| D^{-1/2} \right\| \cdot \left\| A \bchi - \nu D \bchi \right\|
    \leqslant \frac{2}{\sqrt{d}} \cdot \sqrt{2n} 4 d \, \gamma 
    = 8 \sqrt{2nd} \, \gamma.
\end{align}

\noindent
Observe that $\bw_1$ is parallel to $D^{1/2} \bone$ and we have that
\begin{equation}\label{eq:parallelcompbound}
    \left| \bone^\intercal D \bchi \right|
    = \left| \sum_{v \in [2n]} \bchi(v) d_v \right|
    \leqslant (1+\gamma) d n - (1-\gamma) d n
    = 2 n d \, \gamma.
\end{equation}
Hence, if we name $\by$ the component of $D^{1/2} \bchi$ orthogonal to the first
eigenvector, we can write it as
\begin{equation}\label{eq:ortdecompsqrtdchi}
    D^{1/2} \bchi = \frac{\bone^\intercal D \bchi}{\| D^{1/2}\bone \|^2} D^{1/2}\bone
    + \by.
\end{equation}
Thus,
\begin{align}\label{eq:ubdistance}
    \| N \by - \nu \by \| 
    & = \left\| N \left(
    D^{1/2}\bchi - \frac{\bone^\intercal D \bchi}{\| D^{1/2}\bone \|^2} D^{1/2}\bone 
    \right) 
    - \nu \left(
    D^{1/2}\bchi - \frac{\bone^\intercal D \bchi}{\| D^{1/2}\bone \|^2} D^{1/2}\bone 
    \right)
    \right\| \nonumber \\
    & \leqslant \left\| N D^{1/2} \bchi - \nu D^{1/2} \bchi \right\|
    + \frac{\left|\bone^\intercal D \bchi\right|}{\| D^{1/2}\bone \|^2} \,
    \left\|N D^{1/2} \bone - \nu D^{1/2} \bone \right\| \nonumber \\
    & = \left\| N D^{1/2} \bchi - \nu D^{1/2} \bchi \right\|
    + \frac{\left|\bone^\intercal D \bchi\right|}{\| D^{1/2}\bone \|} \,
    \frac{2b}{d} \nonumber\\
    & \leqslant 8 \sqrt{2nd} \, \gamma + 4 \sqrt{2nd} \, \gamma,
\end{align}
where in the last inequality we used~\eqref{eq:sqrtdchibound} 
and~\eqref{eq:parallelcompbound} and the facts that $b \leqslant d/2$ and 
$\left\| D^{1/2}\bone \right\| \geqslant (1/2) \sqrt{2nd}$. 
From~\eqref{eq:ortdecompsqrtdchi} it follows that
\begin{equation}\label{eq:lbnormy}
    \|\by\| \geqslant \left\| D^{1/2} \bchi \right\| 
    - \frac{\bone^\intercal D \bchi}{\left\|D^{1/2} \bone \right\|}
    \geqslant (1 - \gamma) \sqrt{2nd} - 4 \gamma \sqrt{2nd}
    = (1 - 5\gamma) \sqrt{2nd}
    \geqslant (1/2) \sqrt{2nd}.
\end{equation}

\noindent
Now, let us we write $\by$ as a linear combination of the orthonormal
eigenvectors of $N$, $\by = \beta_2 \bw_2 + \cdots + \beta_n \bw_n$ (recall
that $\by^{\intercal}\bw_1=0$ by definition of $\by$ in \eqref{eq:ortdecompsqrtdchi}). 
From \eqref{eq:ubdistance} and \eqref{eq:lbnormy}, it follows that 
\begin{equation}\label{eq:diffnormsquare}
100\gamma^2 \| \by \|^2 \geqslant \left\| N \by - \nu \by \right\|^2
= \left\| \sum_{i = 2}^{n} (\lambda_i - \nu) \beta_i \bw_i \right\|^2
= \sum_{i=2}^n (\lambda_i - \nu)^2 \beta_i^2.
\end{equation}
Moreover, from hypothesis $\lambda_3 < \nu$  we have that
\begin{equation}
    \sum_{i=2}^n (\lambda_i - \nu)^2 \beta_i^2 
    \geqslant \sum_{i=3}^n (\lambda_i - \nu)^2 \beta_i^2
    \geqslant (\lambda_3 - \nu)^2 \sum_{i=3}^n \beta_i^2
    = (\lambda_3 - \nu)^2 \| \by - \beta_2 \bw_2 \|^2.
    \label{eq:part2_first_thesis}
\end{equation}
Thus, by combining together \eqref{eq:diffnormsquare} and \eqref{eq:part2_first_thesis} we get
\begin{equation}
    \| \by - \beta_2 \bw_2 \| \leqslant \frac{10 \, \gamma}{\nu - \lambda_3}\|\by\|
    \label{eq:first_thesis}
\end{equation}
where $\beta_2 = \by^\intercal \bw_2 = \left( D^{1/2}\bchi \right)^\intercal \bw_2$.

\noindent
As for the first thesis of the lemma, observe that if $\lambda_2\geq \nu$ then the first
thesis is obvious. Otherwise, if $\lambda_2 < \nu$, then $(\lambda_2 - \nu)^2 \leqslant
(\lambda_3 - \nu)^2 \leqslant \cdots \leqslant (\lambda_n - \nu)^2$.
Thus, the first thesis follows from \eqref{eq:diffnormsquare} and the fact that
\[
    \sum_{i=2}^n (\lambda_i - \nu)^2 \beta_i^2 
    \geqslant (\lambda_2 - \nu)^2 \sum_{i=2}^n \beta_i^2
    = (\lambda_2 - \nu)^2 \|\by\|^2.
\]

\noindent
As for the second thesis of the lemma, we have 
\begin{align*}
\left\| D^{1/2} \bchi - \beta_2 \bw_2 \right\| 
& = \left\| 
\frac{\bone^\intercal D \bchi}{\| D^{1/2}\bone \|^2} D^{1/2}\bone 
+ \by - \beta_2 \bw_2 
\right\| \\
& \leqslant \frac{\left|\bone^\intercal D \bchi\right|}{\| D^{1/2}\bone \|}
+ \| \by - \beta_2 \bw_2 \|  
\leqslant 4 \,\gamma \sqrt{2nd} \, + \frac{10 \, \gamma}{\nu - \lambda_3}\|\by\|\\
& \leqslant 4 \,\gamma \sqrt{2nd} \, 
+ \frac{20 \, \gamma}{\nu - \lambda_3}{\sqrt{2 nd}} 
\leqslant \frac{44 \, \gamma}{\nu - \lambda_3}{\sqrt{2 nd}},
\end{align*}
where in the last inequality we used that $\by$ is the projection of $D^{\frac 12}\bchi$ on
$D^{\frac 12}\bone$, and thus $\|\by\|\leq \|D^{\frac 12}\bchi\| \leq 2\sqrt{2nd}$.
\end{proof}

\subsection{Proof of Theorem 4.3}
\label{apx-theoremweak}

From Lemma~\ref{lem:clustered} it follows that for every node $u$ at any 
round $t$ we have
\[
    \bx^{(t-1)}(u) - \bx^{(t)}(u) 
    = \alpha_2 \lambda_2^{t-1} (1 - \lambda_2) \left(\bchi(u) + \bz(u)\right)
    + \be^{(t-1)}(u) - \be^{(t)}(u).
\]
Hence, for every node $u$ such that $|\bz(u)| < 1/2$ (we choose $1/2$ here 
for readability sake, however any other constant smaller than $1$ works as well)
it holds that 
$\texttt{sgn}\left( \bx^{(t-1)}(u) - \bx^{(t)}(u) \right) 
= \texttt{sgn}\left(\alpha_2 \bchi(u) \right)$ whenever
\begin{equation}
    \label{eq:clustcondnonreg-apx}
    \left| \frac{1}{2} \alpha_2 \lambda_2^{t-1} (1 - \lambda_2) \right| > 
    \left| \be^{(t-1)}(u) - \be^{(t)}(u) \right|.
\end{equation}
From Lemma~\ref{lem:clustered} we have that 
$\left| \be^{(t)}(u) \right| \leqslant 4 \lambda^t \sqrt{2n}$, 
thus~\eqref{eq:clustcondnonreg-apx} is satisfied for all
\begin{equation}
    t - 1 \geqslant 
    \frac{\log \left(\frac{16 \sqrt{2n}}{|\alpha_2| (1 - \lambda_2)}
    \right)}{\log\left(\lambda_2 / \lambda \right)}.
    \label{eq:time-apx}
\end{equation}
The right-hand side in the above formula is $\mathcal{O}(\log n)$ 
w.h.p., because of the following three points:
\begin{itemize}
    \item From Cheeger's inequality (see e.g. \cite{chung96}) and the fact
        that the graph is connected it follows that $1 - \lambda_2
        \geqslant 1/(2n^4)$;
    \item $\lambda_2 \geqslant (1+\delta)\lambda$ by hypothesis;
    \item It holds $|\alpha_2| \geqslant n^{-c}$ for some large enough
        positive constant $c$ w.h.p., as a consequence of the following
        equations that we prove below:
    \begin{align}
        \Prob{}{ \left| \alpha_2 \right| \leq \frac{1}{n^{c}}}
        = \Prob{}{  \frac{\left|\bw_2^{\intercal}D^{\frac 12}\bx \right|}
            {\left|\bw_2^{\intercal}D^{\frac 12} \bchi\right|}  \leq \frac{1}{n^{c}}}
            \leq \Prob{}{ \left| \bw_2^\intercal D^{1/2} \bx \right| \leq \frac{2 \sqrt{d}}{n^{c-1/2}} }
        \leq \bigO \left( \frac 1 {\sqrt n} \right).
        \label{eq:probability_alpha2}
    \end{align}
        In the first equality of \eqref{eq:probability_alpha2} we used that, by
        definition, $|\alpha_2| = |\bw_2^{\intercal}D^{\frac 12}\bx| /
        |\bw_2^{\intercal}D^{\frac 12} \bchi|$. In the first inequality we used
        that, by the Cauchy-Schwarz inequality, $|\bw_2^{\intercal}D^{\frac 12}
        \bchi| \leq \|D^{\frac 12} \bchi\| \leq 2\sqrt{dn}$. In order to prove the
        last inequality of \eqref{eq:probability_alpha2}, we use that from Lemma
        \ref{lemma:sqrtdchiapprox} it holds
        \begin{equation}
            \left\| D^{1/2} \bchi - \beta_2 \bw_2 \right\|^2 =  
            \left\| D^{1/2} \bchi  \right\|^2 + \left\| \beta_2 \bw_2 \right\|^2 
                - 2\langle D^{1/2} \bchi , \beta_2 \bw_2 \rangle 
            \leqslant 2 \frac{44^2  \, \gamma^2}{(\nu - \lambda_3)^2}\,{nd},
            \label{eq:chi-D_proj}
        \end{equation}
        that is 
        \begin{equation}
            \langle D^{1/2} \bchi , \beta_2 \bw_2 \rangle = 
            \langle D^{1/2} \bchi , \bw_2 \rangle^2 \geq 
            \frac 12 \left( \left\| D^{1/2} \bchi  \right\|^2  
            - 2 \frac{44^2  \, \gamma^2}{(\nu - \lambda_3)^2}\,{nd} \right) \geq \frac{nd}3.
            \label{eq:chi-D_proj2}
        \end{equation}
        Since $\bw_2$ is normalized the absolute value of its entries is at most
        $1$, which toghether with \eqref{eq:chi-D_proj2} implies that at least a
        fraction 12/13 of its entries have an absolute value greater than 1/12.
        Thus, we can apply Lemma \ref{lemma:apxrdm} and prove the last inequality
        of \eqref{eq:probability_alpha2} and, consequently, the fact that
        \eqref{eq:time-apx} is $\bigO(\log n)$.
\end{itemize}

Finally, from Lemma~\ref{lem:clustered} we have 
\[
    \|\bz \| \leqslant \frac{88\, \gamma}{\nu - \lambda_3}\sqrt{2 n} \,.
\]
Thus the number of nodes $u$ with $\bz(u) \geqslant 1/2$ is $\bigO(n {\gamma^2 }/(\nu - \lambda_3)^2)$.

%% file: trunk/apx-random.tex
\section{Stochastic Block Models}\label{sec:apx-random}

\subsection{Regular stochastic block model}

\begin{lemma} \label{ov:lm:regulsbm}
Let $G$ be  a graph sampled from the regular stochastic block model with 
internal and external degrees $a$ and $b$ respectively. W.h.p., it holds that 
\begin{equation} 
\lambda \leqslant \frac{2}{a+b} (\sqrt {a+b-1} + o_n(1))
\end{equation}
\end{lemma}

\begin{proof}
The lemma follows from the general results of Friedman and Kohler \cite{FK14}, 
recently simplified by Bordenave \cite{B15}. If $G$ is a multigraph on $n$ 
vertices, then a {\em random degree $k$ lift} of $G$ is a distribution over 
graphs $G'$ on $kn$ vertices sampled as follows: every vertex $v$ of $G$ is 
replaced by $k$ vertices $v_1,\ldots,v_k$ in $G'$,  every edge $(u,v)$ in $G$ 
is replaced by a random bipartite matching between $u_1,\ldots,u_k$ and 
$v_1,\ldots,v_k$ (if there are multiple edges, each edge is replaced by an 
independently sampled matching) and every self loop over $u$ is replaced by a 
random degree-2 graph over $u_1,\ldots,u_k$ which is sampled by taking a 
random permutation $\pi : \{1,\ldots,k\} \rightarrow  \{1,\ldots,k\}$ and 
connecting $u_i$ to $u_\pi(i)$ for every $i$.

For every lift of any $d$-regular graph, the lifted graph is still $d$-regular, 
and every eigenvalue of the adjacency matrix of the  base graph is still an 
eigenvalue of the lifted graph.
Friedman and Kohler \cite{FK14} prove that, if $d\geq 3$, then with probability 
$1-\bigO(1/k)$ over the choice of a random lift of degree $k$, the new eigenvalues 
of the adjacency matrix of the lifted graph are at most $2\sqrt {d-1} + o_k(1)$ 
in absolute value. Bordenave \cite[Corollary 20]{B15} has considerably 
simplified the proof of Friedman and Kohler; although he does not explicitly state 
the probability of the above event, his argument also 
bound the failure probability by $1/k^{\Omega(1)}$ \cite{Bmail}.

The lemma now follows by observing that the regular stochastic block model is 
a random lift of degree $n$ of the graph that has only two vertices $v_1$ and 
$v_2$, it has $b$ parallel edges between $v_1$ and $v_2$, and it has $a/2$ 
self-loops on $v_1$ and $a/2$ self-loops on $v_2$.

\end{proof}

\subsection{Proof of Lemma~\ref{lemma:irreg}}
\label{ssec:thm_irreg}

    \begin{lemma}
        \label{lemma:irregone} 
        If $a(n),b(n)$ are such that $d:= a+b > \log n$, then w.h.p. (over the
        choice of $G\sim \mathcal{G}_{2n,\frac a{n},\frac b{n}}$), if we let
        $A$ be the adjacency matrix of $G$, then $\|A - B\| \leq \bigO( \sqrt
        {d \log n})$ w.h.p.
    \end{lemma}
    \begin{proof}
        We can write $A-B$ as $\sum_{\{ i,j\}} X^{\{i,j\}}$, where the matrix 
        $X^{\{i,j\}}$ is zero in all coordinates except $(i,j)$ and $(j,i)$, and, 
        in those coordinates, it is equal to $A-B$. Then we see that the matrices 
        $X^{\{ i,j\}}$ are independent, that $\E { X^{\{ i,j\}}} = 0$, that 
        $\| X^{\{ i,j\}}\| \leq 1$, because every row contains at most one non-zero 
        element, and that element is at most 1 in absolute value, and that 
        $\mathbf{E} [ \sum_{\{ i,j \}} ( X^{\{ i,j\}})^2 ]$ is the matrix that is zero 
        everywhere except for the diagonal entries $(i,i)$ and $(j,j)$, in which we
        have $B_{i,i} - B_{i,i}^2$ and $B_{j,j} - B_{j,j}^2$ respectively. It follows
        that 
        \[
            \| \mathbf{E}[ \sum_{\{ i,j \}} ( X^{\{ i,j\}})^2]\| \leq d.
        \] 
        Putting these facts together, and applying the Matrix Bernstein Inequality (see
        Theorem~\ref{thm:matrixBer} in Appendix~\ref{apx:la}) with $t= \sqrt{6 d\log
        n}$, we have
        \[ 
            \pr{} { \| A-B\| \geq \sqrt{9 d \log n} } \leq 2n e^{ -\frac {9d
            \log n}{2d + \frac 23 \sqrt{ 9 d \log n}} } \leq 2n e^{-\frac {9d
            \log n}{4d}} \leq 2n^{-1},
        \]
        where we used $d > \log n$.
    \end{proof}

    \begin{lemma}
        \label{lemma:irregtwo} Let $G$ be a $(2n,d,b,\gamma)$-clustered graph
        such that $\nu = 1- \frac{2b}{d} > 12 \gamma$ and such that its adjacency
        matrix $A$ satisfies $\| A- B\| \leq \gamma d$. Then for every $i\in
        \{ 3,\ldots, {2n} \}$, $|\lambda_i| \leq 4\gamma$ and $\lambda_2 \geq
        (1+\delta) \lambda_3 $ for some constant $\delta>0$.
    \end{lemma}
    \begin{proof}
        The matrix $B$ has a very simple spectral structure: $\bone$ is an eigenvector
        of eigenvalue $d$, $\bchi$ is an eigenvector of eigenvalue $a-b$, and all 
        vectors orthogonal to $\bone$ and to $\bchi$ are eigenvectors of eigenvalue $0$.
        In order to understand the eigenvalues and eigenvectors of $N$, and hence the
        eigenvalues and eigenvectors of $P$, we first prove that $A$ approximates 
        $B$ and that $N$ approximates $(1/d) A$, namely $\| dN - A  \| \leq 3\gamma d$.

        To show that $dN$ approximates $A$ we need to show that $D$ approximates $dI$.
        The condition on the degrees immediately gives us $\| D - dI \| \leq \gamma d$.
        Since every vertex has degree $d_i$ in the range $d\pm \gamma{d}$, then the 
        square root $\sqrt d_i$ of each vertex must be in the range 
        $[\sqrt d - \gamma \sqrt d, \sqrt d + \gamma \sqrt d]$, so we also have 
        the spectral bound:
        \begin{equation} \label{sqrtdbound}
            \| D^{1/2} - \sqrt d I \| \leq \gamma \sqrt d.
        \end{equation}
        We know that $\|D\| \leq d+ \gamma d < 2d$ and that $\| N\|=1$, so from 
        \eqref{sqrtdbound} we get
        \begin{align}
            \| A - dN\| &= \| D^{1/2} N D^{1/2} - dN \| \\
            &\leq \| D^{1/2} N D^{1/2} - \sqrt d  N D^{1/2} \| + \| \sqrt d  N D^{1/2} - dN\| \\
            &= \| (D^{1/2} - \sqrt d I ) \cdot N D^{1/2} \| + \| \sqrt d N \cdot
                (D^{1/2} - \sqrt d I ) \|\\ 
            &\leq \| D^{1/2} - \sqrt d I \| \cdot \| N\| \cdot \|D^{1/2} \| + \sqrt d
                \cdot \|N \| \cdot \| D^{1/2} - \sqrt d I \| \leq 3 \gamma d.
        \label{eq:AvsNdist}
        \end{align}

        \noindent
        By using the triangle inequality and \eqref{eq:AvsNdist} we get
        \begin{equation}
            \| N - (1/d) B \| \leq \| N - (1/d) A\| + (1/d) \cdot \| A-B\| 
            \leq 4 \gamma.
            \label{eq:NapproxB}
        \end{equation}

        \noindent
        Finally, we use Theorem~\ref{thm:eigen_vs_norm} (See Appendix~\ref{apx:la}), which is a
        standard fact in matrix approximation theory: if two real symmetric matrices are
        close in spectral norm then their eigenvalues are close.
        From \eqref{eq:NapproxB} and the fact that all eigenvalues of
        $(1/d) B$ except for the first and second one are $0$,
        for each $i\in \left\{ 3,\dots,2n \right\}$ we have
        \begin{equation}
            |\lambda_i| = | \lambda_i - 0 | \leq \| N - \frac 1d B \| \leq 4\gamma.
            \label{eq:lambda_small}
        \end{equation}
        Similarly, from the fact that
        the second eigenvalue of $(1/d) B$ is $1-2b/d$ we get
        \[
            | \lambda_2 - (1-2b/d) | \leq \| N - \frac 1d B \| \leq 4\gamma,
        \]
        that is, from hypothesis $ \nu > 12 \gamma$ and
        \eqref{eq:lambda_small}, $\lambda_2 \geq (1+\delta) \lambda_3$ for some
        constant $\delta>0$.
        This concludes the proofs of Lemma~\ref{lemma:irregtwo} and Theorem
        \ref{lemma:irreg}.
\end{proof}


\subsection{Proof of Lemma~\ref{lem:main}}
Let $G$ be a randomly-generated graph according to
$\mathcal{G}_{2n,p,q}$ with $a=p{n}$, $b=q{n}$ and $d=a+b$. Recall the
definitions of $A$, $D$, $N$, $P$, $\lambda_i$ and $\bw_i$ ($i\in
\{1,\dots,2n\}$) in Section \ref{sec:prely}, and let $B$ be defined as in
Section \ref{ssec:thm_irreg}. Let us denote with $A_i$ ($i\in \left\{ 1,2
\right\}$) the adjacency matrix of the subgraph of $G$ induced by community
$V_i$, with $A_B=\left\{A_{u,v-n}\right\}_{u\in V_1, v\in V_2}$ the matrix whose
entry $(i,j)$ is $1$ iff there is an edge between the $i$-th node of $V_1$ and the $j$-th
node of $V_2$, then
\[
    A= \begin{pmatrix}
            A_1 & A_B \\
            A_B^{\intercal} & A_2
        \end{pmatrix}.
\] 
We need the following technical lemmas. 

\begin{lemma}\label{lem:versh_lemma}
    If $d>5\log n$ then for some positive constant $\Cr{spect_const}$ it holds
    $\| A - B \| \leq \Cl[small]{spect_const}\sqrt{d}$ w.h.p. 
\end{lemma}
\begin{proof}
    The lemma directly follows from Theorem 2.1 in \cite{le_concentration_2015}
    with $d'=2d$ and the observation that, from the Chernoff bounds, all
    degrees are smaller than $2d$ w.h.p.
\end{proof}

\begin{lemma}
    \label{lem:N_vs_B}
    If $d>5 \log n$ then for some constant
    $\Cl[small]{N_vs_B_bound}>0$ it holds w.h.p.
    \begin{equation}
        \| dN - B \| 
            \leq \Cr{N_vs_B_bound}\sqrt{d}.
        \label{eq:N_vs_B}
    \end{equation}
\end{lemma}

\noindent
The idea for proving Lemma~\ref{lem:N_vs_B} is to use the triangle inequality 
to upper bound $\| dN - B \|$ in terms of $\| A - B \|$, which we can bound with 
Lemma~\ref{lem:versh_lemma}, and $\|B-1/d D^{1/2} B D^{1/2}\|$, which we can 
upper bound by bounding $\| \sqrt{d} \bone - D^{1/2} \bone\|$ and 
$\|\sqrt{d}\bchi - D^{1/2} \bchi\|$ where $\bone$ and $\bchi$ are the 
eigenvector corresponding to the only two non-zero eigenvalues of $B$.  
The complete proof of Lemma~\ref{lem:N_vs_B} is deferred to 
Section~\ref{apx:dN_vs_B}.
As for the required bound on 
$\| \sqrt{d} \bone - D^{1/2} \bone\| = \| \sqrt{d}\bchi - D^{1/2} \bchi\| =
\sum_{j\in V} | \sqrt d - \sqrt {d_j} |^2$, we provide it in the following 
lemma, whose proof is also deferred to Section~\ref{apx:sum_of_deg}. 

\begin{lemma} 
    \label{lem:sumdeg} 
    If $5\log n < d< n^{\frac 13-\Cr{tight}}$ for any constant
    $\Cl[small]{tight}>0$, it holds w.h.p.
    \begin{align}
        \label{eq:sum_of_deg}
        & \sum_{j\in V} | \sqrt d - \sqrt {d_j} |^2 \leq 2n \text{ and }\\
        \label{eq:sum_of_deg3}
        &\sum_{j\in V}  | d - {d_j} |^2 \leq 2dn.
    \end{align}
\end{lemma}

\noindent
By combining Lemma~\ref{lem:N_vs_B} and Theorem~\ref{thm:eigen_vs_norm} we have
$|\lambda_i - \lambda_i'|\leq \| N - d^{-1}B \| = \bigO(1/\sqrt{d})$, where
$\lambda_1'= 1$, $\lambda_2'= 1-2b/d$ and $\lambda_i'= 0$ for $i\in \left\{
3,\dots,2n \right\}$ are the eigenvalues of $d^{-1}B$. This proves the first
two part of Lemma~\ref{lem:main}.

As for the third part, let us write $\bw_2 = \bw_{\bone} + \bw_{\bchi} + \bw_{\perp}$ 
where $\bw_{\bone}$ and $\bw_{\bchi}$ are the projection of $\bw_2$ on $\bone$ and $\bchi$
respectively, and $\bw_{\perp}$ is the projection of $\bw_2$ on the space orthogonal to 
$\bone$ and $\bchi$. 

Observe that the only non-zero eigenvalues of $(1/d)B$ are $1$ and $(a-b)/d$. 
Thus, from Lemma \ref{lem:N_vs_B} and the Davis-Kahan theorem (Theorem \ref{thm:daviskahan}) with 
$M_1 = N$, $M_2 = \frac{1}{d} B$, $t=\lambda_2$, $\bx = \bw_2$ and $\delta = \lambda_2/2$, we get
\begin{equation}
    \| \bw_{\perp}\| 
        \leq \frac{4}{\lambda_2\pi} \Big\|N - \frac{1}{d}B\Big\|
        \leq \bigO\left(\frac{1}{\sqrt{d}\lambda_2}\right) 
        = \bigO\left(\frac{\sqrt{d}}{a-b}\right).
    \label{eq:davis_kahan}
\end{equation}

As for $\bw_{\bone}$, we know that $\langle \bw_2, D^{-1/2}\bone\rangle=0$, thus
\begin{align}
    \label{eq:w_bone_bound}
    \|\bw_{\bone}\| 
    = \frac{1}{\sqrt{2n}}\langle \bw_2, \bone-d^{-\frac{1}{2}}D^{\frac{1}{2}}\bone\rangle 
    \leq \frac{1}{\sqrt{2n}}\|\bw_2\| \|\bone-d^{-\frac{1}{2}}D^{\frac{1}{2}}\bone\| 
    \leq \frac{1}{\sqrt{d}},
\end{align}
where in the last inequality we used Lemma \ref{lem:sumdeg}.

By the law of cosines and the fact that $\sqrt{1-x}\geq 1-x$ for $x\in\left[
0,1 \right]$ we have that 
\begin{align}
    \label{eq:proj_of_w_on_chi}
    \Big\| \bw_2 - \frac{1}{\sqrt{2n}}\bchi \Big\|^2 
        &= \|\bw_2\|^2 + \Big\|\frac{1}{\sqrt{2n}}\bchi\Big\|^2 
            - 2\langle \bw_2, \frac{1}{\sqrt{2n}}\bchi \rangle = 2 - 2\|\bw_{\bchi}\|\\
        &= 2-2\sqrt{ 1 - \|\bw_{\bone}\|^2 + \|\bw_{\perp}\|^2 }
            \leq 2\left( \|\bw_{\bone}\|^2 + \|\bw_{\perp}\|^2 \right)
            = \bigO \left( \frac{{d}}{(a-b)^2} \right),
\end{align}
where in the last inequality we used \eqref{eq:davis_kahan} and
\eqref{eq:w_bone_bound}. \eqref{eq:proj_of_w_on_chi} implies that, with the
exception of a set $S$ of at most $\bigO(nd/(a-b)^2)$ nodes, we have
\begin{equation}
    \label{eq:pre_third_thesis_lemma}
    \left|\sqrt{2n}\bw_2(i) - \bchi(i)\right| \leq \frac{1}{201},
\end{equation}
for each $i \in V/S$. 
From the Chernoff bound, we also have that w.h.p. $\sqrt{d/d_i} = 1\pm 1/201$. 
Thus, \eqref{eq:pre_third_thesis_lemma} and the last fact imply that for each $i \in V/S$
it holds w.h.p. 
\begin{equation}
    \label{eq:third_thesis_lemma}
    \left|\sqrt{2nd}D^{-\frac{1}{2}}\bw_2(i) - \bchi(i)\right| \leq \frac{1}{100},
\end{equation}
concluding the proof.
\qed

\begin{remark}
    After looking at Lemma~\ref{lem:main}, one may wonder whether it could be
    enough to generalize Definition \ref{def:epsreggraph} to include
    ``quasi-$(2n,d,b,\gamma)$-clustered graph'', i.e. graphs that are
    $(2n,d,b,\gamma)$-clustered except for a small number of nodes which may
    have a much higher degree. In fact, this would be rather surprising: This
    higher-degree nodes may connect to the other nodes in such a way that would
    greatly perturb the eigenvalues and eigenvectors of the graph. In
    $\planted$, besides the fact that the nodes with degree much larger than
    $d$ are few, it is also crucial that they are connected in a
    \emph{non-adversarial} way, i.e. randomly.
\end{remark}

\subsection{Proof of Lemma~\ref{lem:N_vs_B}}
\label{apx:dN_vs_B}
A simple application of the Chernoff bound and the union bound shows that 
w.h.p. 
\begin{equation}
        \sqrt{d}\|D^{-1/2}\|\leq 1 + \bigO\left(\sqrt{\frac{\log n}d}\right),
        \label{eq:sqrtd_times_Dsqrt}
    \end{equation}hence
    \begin{align}
        \| dN - B \|&= \| (\sqrt{d}D^{-1/2})A(\sqrt{d}D^{-1/2})-B \| \\
            &\leq \|\sqrt{d}D^{-1/2}\| \left\|A - \frac1{\sqrt{d}} D^{1/2} B
                \frac1{\sqrt{d}} D^{1/2} \right\| \|\sqrt{d}D^{-1/2}\|\\
            & \leq \left\| A - \frac1{d} D^{1/2} B D^{1/2} \right\| \| \sqrt{d}D^{-1/2} \|^2\\
            &\leq \left( \|A-B\| + \left\| B - \frac1{d} D^{1/2} B D^{1/2} \right\| \right)
                \left( 1+ \bigO\left( \sqrt{\frac{\log n}{d}} \right) \right).
        \label{eq:proof_N_vs_B}
    \end{align}
    Thanks to Lemma~\ref{lem:versh_lemma}, it holds $\| A - B \| = \bigO(\sqrt{d})$.
    Hence, in order to conclude the proof, it remains to show that $\| B - {d^{-1}}
    D^{1/2} B D^{1/2} \| = \bigO(\sqrt{d})$. 
    We do that by observing that 
    \begin{align}
        \label{eq:B_vs_DB}
        \left\| B - \frac1{d} D^{1/2} B D^{1/2} \right\| 
            \leq \left\| B - \frac1{\sqrt{d}} B D^{1/2} \right\| 
            + \left\| \frac1{\sqrt{d}} B D^{1/2} - \frac1{d} D^{1/2} B D^{1/2} \right\|,
    \end{align}
    and by upper-bounding the two terms on the right hand side.
The two only non-zero eigenvalues of $B$ are $a+b$ and $a-b$, with corresponding 
eigenvectors $(2n)^{-\nfrac 12} \, \bone$ and $(2n)^{-\nfrac 12} \, \bchi$, 
therefore
we can write $B = d/(2n) \, \bone \bone^{\intercal} + (a-b)/(2n) \, \bchi \bchi^{\intercal}$,
which implies that 
    \begin{align}
        \label{eq:B_vs_DB_rewritten}
        B - \frac1{\sqrt{d}} B D^{1/2} = 
            \frac{\sqrt{d}}{2n} \, \bone 
                \, (\sqrt{d}\, \bone - D^{1/2}\, \bone)^{\intercal} 
            + \frac{ a-b }{\sqrt{d}\, 2n} \, \bchi 
                \, (\sqrt{d}\, \bchi - D^{1/2}\, \bchi)^{\intercal}.
    \end{align}
    It follows that, for an arbitrary unitary vector $\bx$ it holds
    \begin{align}
        \label{eq:B_vs_DB_bound}
        \left\| \left(B - \frac1{\sqrt{d}} B D^{1/2}\right)\bx \right\|
        & \leq \left\|\frac{\sqrt{d}}{2n} \, \bone 
                \, (\sqrt{d}\, \bone - D^{1/2}\, \bone)^{\intercal} \bx \right\|
            + \left\|\frac{ a-b }{\sqrt{d}\, 2n} \, \bchi 
                \, (\sqrt{d}\, \bchi - D^{1/2}\, \bchi)^{\intercal} \bx \right\|\\
        & = \frac{\sqrt{d}}{2n} \, \left\|\bone\right\| 
                \, |(\sqrt{d}\, \bone - D^{1/2}\, \bone)^{\intercal} \bx| 
            + \frac{ a-b }{\sqrt{d}\, 2n} \, \left\|\bchi \right\|
                \, |(\sqrt{d}\, \bchi - D^{1/2}\, \bchi)^{\intercal} \bx|\\
                & \leq \frac{\sqrt{d}}{\sqrt{2n}} 
                \, \left\|\sqrt{d}\, \bone - D^{1/2}\, \bone\right\| \cdot \left\|\bx\right\| 
            + \frac{ a-b }{\sqrt{2dn}} 
                \, \left\|\sqrt{d}\, \bchi - D^{1/2}\, \bchi\right\| \cdot \left\|\bx\right\|
        \leq 2\sqrt{d},
    \end{align}
    where we used the triangle inequality, the fact that $\|\bone\|=\|\bchi\|=\sqrt{2n}$,
    the Cauchy-Schwartz inequality, Lemma \ref{lem:sumdeg} and $a-b<d$.
    As for the other term on the r.h.s. of 
    \eqref{eq:B_vs_DB}, we have that w.h.p.
    \begin{align}
        \label{eq:B_vs_DB_last}
        \left\| \frac1{\sqrt{d}} B D^{1/2} - \frac1{d} D^{1/2} B D^{1/2} \right\|
            \leq \left\| B - \frac1{\sqrt{d}}  D^{1/2} B \right\| \frac1{\sqrt{d}}\|D^{1/2}\|
            \leq 2\sqrt{d}\left( 1+ \bigO\left( \sqrt{\frac{\log n}{d}} \right) \right),
    \end{align}
    where in the last inequality we used \eqref{eq:sqrtd_times_Dsqrt} and that
    for any matrix $M$ it holds $\|M\| = \|M^{\intercal}\|$. Finally,
    \eqref{eq:B_vs_DB_bound} and \eqref{eq:B_vs_DB_last} togheter implies the
    desired upper bound on \eqref{eq:B_vs_DB} and thus \eqref{eq:proof_N_vs_B},
    concluding the proof.
\qed

\subsection{Proof of Lemma~\ref{lem:sumdeg}}
\label{apx:sum_of_deg}

Each degree $d_i$ has the distribution of a sum of $n$ Bernoulli random
variables of expectation $p$ plus a sum of $n$ Bernoulli random variables
of expectation $q$. Thus, each $d_i$ satisfies $\E {d_i} = d$ and $\Var {}
{d_i } \leq d$. 

First, we consider the random variables $| d -
d_j|^2$. Their expectation is $\E {| d-d_j|^2} \leq d$ (the variance of the
random variable $d_j$). Let $e_{u,v}$ is the variable that is 1 iff the edge
$(u,v)$ is included in the graph. 
Observe that
\begin{align}
    | d- d_j|^4 
        & = | d - \sum_{v\in V}e_{j,v} |^4 = | a - \sum_{v\in V_i}e_{j,v} 
            + b - \sum_{v\in V_{3-i}}e_{j,v} |^4 \\
        & = | a - \sum_{v\in V_i}e_{j,v} |^4
        + | b - \sum_{v\in V_{3-i}}e_{j,v} |^4
        + 6 | a - \sum_{v\in V_i}e_{j,v} |^2 \, | b - \sum_{v\in V_{3-i}}e_{j,v} |^2\\
        &\qquad + 4 (a - \sum_{v\in V_i}e_{j,v})  \, ( b - \sum_{v\in V_{3-i}}e_{j,v} )^3
        + 4 (a - \sum_{v\in V_i}e_{j,v})^3  \, ( b - \sum_{v\in V_{3-i}}e_{j,v} ),
    \label{eq:horrible}
\end{align}
and 
\begin{align}
    \E {(a - \sum_{v\in V_i}e_{j,v})^3 ( b - \sum_{v\in V_{3-i}}e_{j,v} )}
        = \E {(a - \sum_{v\in V_i}e_{j,v})^3 } \E{ ( b - \sum_{v\in V_{3-i}}e_{j,v} )}=0 ,\\
    \E {(a - \sum_{v\in V_i}e_{j,v}) ( b - \sum_{v\in V_{3-i}}e_{j,v} )^3}
        = \E {(a - \sum_{v\in V_i}e_{j,v}) } \E{ ( b - \sum_{v\in V_{3-i}}e_{j,v} )^3}.
    \label{eq:zero_exp}
\end{align}
Hence, since the fourth central moment of a binomial with parameters $n$ and
$p$ is $np(1-p)^4+np^4(1-p)+3n(n-1)p^2(1-p)^2 \leq 4(np)^2$, if we let $i\in
\left\{ 1,2 \right\}$ be the index of the community of $j$ we have that the
expectation of the square of $| d-d_j|^2$ (which is the fourth central moment
of $d_j$) is  
\begin{align}
    \E{| d- d_j|^4 }
        & = \E {| a - \sum_{v\in V_i}e_{j,v} |^4} 
        + \E {| b - \sum_{v\in V_{3-i}}e_{j,v} |^4}\\
        &+ 6 \E {| a - \sum_{v\in V_i}e_{j,v} |^2 } \E{ | b - \sum_{v\in V_{3-i}}e_{j,v} |^2}
         \leq 4a^2 + 4b^2 + 6ab \leq 4 d^2.
    \label{eq:fourth_of_d}
\end{align}
In order to apply Chebyshev's inequality, we need to bound the variance of
$\sum_j | d - d_j|^2$. As for the second moment of their sum, we have
\begin{align}
    \mathbf{E}[(\sum_{i}|d-d_j|^2)^2] 
    &= \sum_{i}\mathbf{E}[|d-d_j|^4]
        + 2\sum_{ 1 \leq i < j \leq 2n } \mathbf{E}[|d-d_i|^2\cdot|d-d_j|^2]\\
    &\leq 8d^2 n + 2\sum_{ 1 \leq i < j \leq 2n } \mathbf{E}[|d-d_i|^2\cdot|d-d_j|^2].
    \label{eq:calc2}
\end{align}
To upper bound the terms $\mathbf{E}[|d-d_i|^2\cdot|d-d_j|^2]$, since the
stochastic dependency between $d_i$ and $d_j$ is due only to the edge $(i,j)$,
let us write 
\[
    d_i=\sum_{u\in N(i)}e_{i,u} = e_{i,j} + \sum_{u\in N(i)/\{j\}}e_{i,u} = e_{i,j} + d^{(j)}_i,
\]
where $d^{(j)}_i$ is the sum of all the edges incident to $i$ except for $(i,j)$.
We have 
\begin{align}
    \label{eq:calc2b}
    |d-d_i|^2\cdot|d-d_j|^2 
    &= |d-d^{(j)}_i+ e_{i,j} |^2\cdot|d-d^{(i)}_j+e_{i,j}|^2 \\
    &{=} (|d-d^{(j)}_i|^2+ e_{i,j} + 2 e_{i,j}(d-d^{(j)}_i))
                      (|d-d^{(i)}_j|^2+ e_{i,j} + 2 e_{i,j}(d-d^{(i)}_j))\\
    &=|d-d^{(j)}_i|^2|d-d^{(i)}_j|^2+ e_{i,j}|d-d^{(i)}_j|^2 + 2 e_{i,j}(d-d^{(j)}_i)|d-d^{(i)}_j|^2\\
        &+ |d-d^{(j)}_i|^2 e_{i,j} + e_{i,j} + 2 e_{i,j}(d-d^{(j)}_i)\\
        &+  2 e_{i,j}(d-d^{(i)}_j)|d-d^{(j)}_i|^2+  2 e_{i,j}(d-d^{(i)}_j) + 4 e_{i,j}(d-d^{(j)}_i)(d-d^{(i)}_j),
\end{align}
where we used that, since $e_{i,j}$ is an indicator variable, it holds $e^{2}_{i,j}=e_{i,j}$.
Taking the expectation of \eqref{eq:calc2b} we thus get
\begin{align}
    &\mathbf{E}[ |d-d_i|^2\cdot|d-d_j|^2 ] \\
    &=\mathbf{E}[ |d-d^{(j)}_i|^2|d-d^{(i)}_j|^2+ e_{i,j}|d-d^{(i)}_j|^2 + 2
            e_{i,j}(d-d^{(j)}_i)|d-d^{(i)}_j|^2\\
        &\qquad+ |d-d^{(j)}_i|^2 e_{i,j} + e_{i,j} + 2 e_{i,j}(d-d^{(j)}_i)\\
        &\qquad+  2 e_{i,j}(d-d^{(i)}_j)|d-d^{(j)}_i|^2+  2 e_{i,j}(d-d^{(i)}_j) + 4
            e_{i,j}(d-d^{(j)}_i)(d-d^{(i)}_j)]\\
    &=\mathbf{E}
            [|d-d^{(j)}_i|^2]\mathbf{E}[|d-d^{(i)}_j|^2]+\mathbf{E}[e_{i,j}]\mathbf{E}[
            |d-d^{(i)}_j|^2 ] + 2 \mathbf{E}[ e_{i,j} ]\mathbf{E}[
            (d-d^{(j)}_i) ]\mathbf{E}[ |d-d^{(i)}_j|^2 ]\\
        &\qquad+ \mathbf{E}[ e_{i,j} ] \mathbf{E}[ |d-d^{(j)}_i|^2 ] +
            \mathbf{E}[ e_{i,j} ] + 2 \mathbf{E}[ e_{i,j} ]\mathbf{E}[
            (d-d^{(j)}_i) ]\\
        &\qquad+ 2 \mathbf{E}[ e_{i,j}]\mathbf{E}[ (d-d^{(i)}_j) ]\mathbf{E}[
            |d-d^{(j)}_i|^2 ]+  2 \mathbf{E}[ e_{i,j} ]\mathbf{E}[ (d-d^{(i)}_j) ]\\
        &\qquad + 4 \mathbf{E}[ e_{i,j} ]\mathbf{E}[ (d-d^{(j)}_i) ]\mathbf{E}[
            (d-d^{(i)}_j) ]\\
    &\leq \mathbf{E}
            [|d-d^{(j)}_i|^2]\mathbf{E}[|d-d^{(i)}_j|^2] + \frac {d^2}n 
            + 2 \frac {d^3}{n^2} + \frac {d^2}n +
             \frac {d}n + 2  \frac {d^2}{n^2} +  2 \frac {d^3}{n^2} +  2 \frac
             {d^2}{n^2}  + 4 \frac {d^3}{n^3} \\
    & \leq \mathbf{E}
            [|d-d^{(j)}_i|^2]\mathbf{E}[|d-d^{(i)}_j|^2] + 15\frac {d^2}n, 
    \label{eq:calc2c}
\end{align}
where in the inequalities we used that $\mathbf{E}[ e_{i,j} ]\leq d/n$, that 
\[
    \mathbf{E} [d-d^{(j)}_i] \leq \mathbf{E}[ e_{i,j} ] + \mathbf{E}
    [\sum_{u\in N(i)/\{j\}}\mathbf{E}[e_{i,u}] -d^{(j)}_i] \leq \frac dn ,
\]
and that 
\begin{equation}
    \mathbf{E} [|d-d^{(j)}_i|^2] \leq \mathbf{E}[ e_{i,j} ] + \mathbf{E}
    [|d-\mathbf{E}[ e_{i,j} ]-d^{(j)}_i|^2] \leq \frac dn + d-1 \leq d. 
        \label{eq:uppersquare_d}
\end{equation}
By combining \eqref{eq:calc2} and \eqref{eq:calc2c} we get
\begin{align}
    \label{eq:calc2_final}
    \mathbf{E}[(\sum_{i}|d-d_j|^2)^2] 
    \leq 8d^2 n + 2 \sum_{1\leq i<j \leq
        2n}\mathbf{E}[|d-d^{(j)}_i|^2]\mathbf{E}[|d-d^{(i)}_j|^2] + 60 {d^2}n, 
\end{align}

As for the square of the average, we have
\begin{align}
    (\mathbf{E}[\sum_{i}|d-d_i|^2])^2 &=
    \sum_{i}\mathbf{E}[|d-d_i|^2]^2 + 2\sum_{i\neq j} \mathbf{E}[|d-d_i|^2]\mathbf{E}[
     |d-d_j|^2]\\
     &\geq 2\sum_{1\leq i<j \leq 2n} \mathbf{E}[|d-d_i|^2]\mathbf{E}[ |d-d_j|^2],
    \label{eq:calc2d}
\end{align}
and
\begin{align}
    &\mathbf{E}[|d-d_i|^2]\mathbf{E}[ |d-d_j|^2]\\
    &= \mathbf{E}[|d-d^{(j)}_i-e_{i,j}|^2]\mathbf{E}[ |d-d^{(i)}_j-e_{i,j}|^2]\\
    &= (\mathbf{E}[|d-d^{(j)}_i|^2] + \mathbf{E}[e_{i,j}] -
        2\mathbf{E}[e_{i,j}]\mathbf{E}[(d-d^{(j)}_i)])
        \cdot (\mathbf{E}[|d-d^{(i)}_j|^2] \\
    &\qquad+ \mathbf{E}[e_{i,j}] -
        2\mathbf{E}[e_{i,j}]\mathbf{E}[(d-d^{(i)}_j)])\\
    &\geq (\mathbf{E}[|d-d^{(j)}_i|^2] - 2\mathbf{E}[e_{i,j}]\mathbf{E}[(d-d^{(j)}_i)])
        \cdot (\mathbf{E}[|d-d^{(i)}_j|^2] - 2\mathbf{E}[e_{i,j}]\mathbf{E}[(d-d^{(i)}_j)])\\
    & \geq \mathbf{E}[|d-d^{(j)}_i|^2]\mathbf{E}[|d-d^{(i)}_j|^2]-4\frac{d^3}{ n^2 },
    \label{eq:calc2e}
\end{align}
where we used, again, that $\mathbf{E}[ e_{i,j} ]\leq d/n$ and that 
$ \mathbf{E} [|d-d^{(j)}_i|^2] \leq d$ (see \eqref{eq:uppersquare_d}).

Combining \eqref{eq:calc2_final} and \eqref{eq:calc2e} together we get
\begin{align}
    \mathrm{Var}[\sum_{i}|d-d_i|^2] &=\mathbf{E}[(\sum_{i}|d-d_i|^2)^2]-
        \mathbf{E}[\sum_{i}|d-d_i|^2]^2\\
    &\leq 8d^2n+60d^2n + 16{d^3} 
        = 84 d^2n
    \label{eq:calcvar}
\end{align}

Finally, by Chebyshev's inequality we have 
\[ 
    \Pr { \left[ \sum_{j}  |  d -  {d_j} |^2 > 2dn \right]} \leq \frac {21} n,
\]
which proves the second part of the lemma. 

We now consider the sum of the variables $| \sqrt d- \sqrt{d_j} |^2$. 
We have
\begin{align} 
    \sum_{j\in V} | \sqrt d- \sqrt{d_j} |^2 
    &= \sum_{i\in V}d + \sum_{i\in V}d_i- 2\sqrt d \cdot \sum_{j\in V}\sqrt {d_j}\\
    &\leq 2dn + \sum_{i\in V}d_i- 2\sqrt d \cdot \sum_{j\in V}\sqrt {d_j}.
    \label{eq:degvar}  
\end{align}
From the Chernoff bound we have that for some positive constant
$\Cl[small]{cb_tight}$ it holds w.h.p. 
\begin{equation}
    \label{eq:cb_sumofdeg}
    \sum_{j\in V} {d_j} 
    = \sum_{\substack{u,v\in V\\u\neq v}} 2e_{u,v} + \sum_{u\in V} e_{u,v} 
        \leq 2dn + \Cr{cb_tight}\sqrt{dn\log n} \leq 4dn + {n},
\end{equation}
where we are using the hypothesis $d=o(n/\log n)$.
We will now prove that 
\begin{equation}
    \label{eq:sumofsq_goal}
    \sum_{j\in V}\sqrt {d_j} \geq 2n\sqrt{d} - \frac{n}{\sqrt{d}},
\end{equation}
which together with \eqref{eq:degvar} implies that 
\begin{equation}
    \label{eq:goal_of_lemma_sumdeg}
    \sum_{j\in V} | \sqrt d- \sqrt{d_j} |^2 \leq 4n,
\end{equation}
concluding the proof of the lemma.

Observe that if $x\geq 0$, we have
\[ 
    \sqrt x \geq 1 + \frac {x-1}2 - \frac {(x-1)^2} 2 
\]
so that if $X$ is a non-negative random variable of expectation 1 we
have\footnote{This argument is due to Ori Gurel-Gurevich (see \cite{121424}).}
\[ 
    \E { \sqrt {X}} \geq 1 - \frac {\Var {}{X}}{2}. 
\]
By applying the above inequality to $d_j/d$ we get
\[ 
    \E{ \sqrt {\frac {d_j}{d}} } \geq 1 - \frac{\Var{}{\frac {d_j}{d} } } 2  =
    1 - \frac { \Var{} {d_j }} {2d^2} \geq 1 - \frac 1 {2d} 
\] 
and
\begin{equation} 
    \label{eq:degmean} 
    \E { \sqrt { d_j } } \geq \sqrt d - \frac 1 {2 \sqrt d}.
\end{equation}
We will show that $\sum_{j \in V} \sqrt{d_j}$ is concentrated around its expectation by 
using Chebyshev's inequality\footnote{A stronger bound which doesn't require the hypothesis
    $d\leq n^{1/3-\Cr{tight}}$ may be obtained with some concentration techniques compatible with the
    stochastic dependence among the $\sqrt{d_j}$s.}. In order to do that, we will bound
    their covariance as 
    \begin{equation}
        \E{ \sqrt{d_i d_j}} - \E{\sqrt{d_i}}\E{\sqrt{d_j}} \leq \frac{8d^2}n.
        \label{eq:bound_on_cov_sqrt}
    \end{equation}
    By the law of total probability 
    \begin{equation}
        \label{eq:total_prob_for_cov}
        \E{\sqrt{d_i}} = \Pr{e_{i,j}}\E{\sqrt{d_i^{(j)}+1}}+(1-\Pr{e_{i,j}})\E{\sqrt{d_i^{(j)}}}
    \end{equation}
    and
    \begin{equation}
        \label{eq:total_prob_for_cov_double}
        \E{\sqrt{d_j d_i}} = \Pr{e_{i,j}}\E{\sqrt{ d_i^{(j)}+1 }}\E{\sqrt{d_j^{(i)}+1}}
            +(1-\Pr{e_{i,j}})\E{\sqrt{d_j^{(i)}}} \E{\sqrt{d_i^{(j)}}},
    \end{equation}
    which imply that 
    \begin{align}
        &\E{ \sqrt{d_i d_j}} - \E{\sqrt{d_i}}\E{\sqrt{d_j}}\\
        & = \Pr{e_{i,j}}\E{\sqrt{ d_i^{(j)}+1 }}\E{\sqrt{d_j^{(i)}+1}}
            +(1-\Pr{e_{i,j}})\E{\sqrt{d_j^{(i)}}} \E{\sqrt{d_i^{(j)}}}\\
        &\quad - \Pr{e_{i,j}}^2 \E{\sqrt{d_j^{(i)}+1}}\E{\sqrt{d_i^{(j)}+1}}
            -\Pr{e_{i,j}}(1-\Pr{e_{i,j}}) \E{\sqrt{d_j^{(i)}}}\E{\sqrt{d_i^{(j)}+1}}\\
        &\quad - \Pr{e_{i,j}}(1-\Pr{e_{i,j}}) \E{\sqrt{d_j^{(i)}+1}}\E{\sqrt{d_i^{(j)}}}
            -(1-\Pr{e_{i,j}})^2 \E{\sqrt{d_j^{(i)}}}\E{\sqrt{d_i^{(j)}}}\\
        & = p(1-p)\Big(\E{\sqrt{ d_i^{(j)}+1 }}\E{\sqrt{d_j^{(i)}+1}}
            +\E{\sqrt{d_j^{(i)}}} \E{\sqrt{d_i^{(j)}}} + \E{\sqrt{d_j^{(i)}}}\E{\sqrt{d_i^{(j)}+1}}\\
        &\quad + \E{\sqrt{d_j^{(i)}+1}}\E{\sqrt{d_i^{(j)}}}\Big) 
            \leq \frac{ 8d^2}{n},
        \label{eq:bound_on_cov_calc}
    \end{align}
    where in the last inequality we used that by the Chernoff bound w.h.p. it
    holds $\E{\sqrt{d_i^{(j)}}} < \sqrt{2d}$, and that $p(1-p)<p<d/n$.
    From \eqref{eq:bound_on_cov_calc} it then follows that 
    \begin{equation}
        \Vr{\sum_{j\in V}\sqrt{d_j}} \leq 2nd + 32d^2n < \frac{n^{2}}{dn^{\Cr{tight}}}.
        \label{eq:var_of_sqrtd}
    \end{equation}
    Finally, by combining \eqref{eq:var_of_sqrtd} and \eqref{eq:degmean} with
    Chebyshev's inequality we get 
    \begin{equation}
        \Pr{\sum_{j\in V}\sqrt{d_j}<2n\sqrt d - \frac n{\sqrt d}} 
            \leq \Pr{ \Big|\sum_{j\in V}\sqrt{d_j} 
                - \mathbf{E}\Big[ \sum_{j\in V}\sqrt{d_j} \Big] \Big| 
                > \frac{n}{\sqrt{d}}} \leq \frac{1}{n^{\Cr{tight}}}.
        \label{eq:cheb_sqrtd}
    \end{equation}
\qed

\subsection{Proof of Theorem~\ref{thm:tight_bounds}}
For any vector $\bx$, we can write
\begin{align}
\bx^{(t)} = P^t \bx = \sum_{i = 1}^{2n} a_i \lambda_i^t D^{-1/2} \bw_i
= \alpha_1 \bone + a_2 \lambda_2^t D^{-1/2} \bw_2 + \be^{(t)},
\label{eq:tight_decomp}
\end{align}
where $\alpha_1 = \bone^{\intercal}D\bx/\|D^{1/2}\bone\|$ and
$\| \be^{(t)} \| \leq 4 \lambda^{t}\|\bx\|$.

From Lemma \ref{lem:main} (Claim 3) we have that for at
least $2n-\bigO(n d / (a-b)^2)$ entries $i$ of $D^{-1/2} \bw_2$, we
get $| \sqrt{2nd}( D^{-1/2} \bw_2 )(i) - \bchi(i) | \leq \frac 1 {100}$,
that is
\begin{align}
    ( D^{-1/2} \bw_2 )(i) &\geq  \frac{99}{100\sqrt{2nd}} \quad \text{ if }i\in V_1 \cap S \text{ and }\\
    ( D^{-1/2} \bw_2 )(i) &\leq -\frac{99}{100\sqrt{2nd}} \quad \text{ if }i\in V_2 \cap S.
    \label{eq:diff_in_S}
\end{align}
Thus, we get 
\begin{align}
    \left|\bx^{(t)} - \bx^{(t-1)}\right| 
    &= \left|a_2 \lambda_2^{t-1}(\lambda_2 - 1) D^{-1/2} \bw_2 + \be^{(t)} + \be^{(t-1)}\right|\\
    &\leq \left|a_2 \lambda_2^{t-1}(\lambda_2 - 1) D^{-1/2} \bw_2\right| 
    + \left|\be^{(t)}-\be^{(t-1)}\right|
    \label{eq:tight_clustering}
\end{align}
and, when 
$t - 1 \geqslant 
\log \left(\frac{16 \sqrt{2n}}{|a_2| (1 - \lambda_2)}
\right) / \log\left(\frac{\lambda_2}{\lambda} \right),$
from \eqref{eq:tight_clustering} it follows that
\begin{align}
    (\bx^{(t)}-\bx^{(t-1)})(i) &\geq  \frac{99}{200\,\sqrt{2nd}}\,a_2
    \lambda_2^{t-1}(\lambda_2 - 1) \quad \text{ if }i\in V_j \cap S \text{
    and }\\
    (\bx^{(t)}-\bx^{(t-1)})(i) &\leq -\frac{99}{200\,\sqrt{2nd}}\,a_2
    \lambda_2^{t-1}(\lambda_2 - 1) \quad \text{ if }i\in V_{3-j} \cap S.
    \label{eq:tight_diff_in_S}
\end{align}
either for $j=1$ or for $j=2$.
Since $|S|>n-\bigO(n d / (a-b)^2)$, we thus get a $\bigO(d / (a-b)^2)$-weak
reconstruction.
\qed

%% file: trunk/apx-more.tex
\section{More communities}
\label{sec:apx-more}

Recall the definition of negative and positive type in Section
\ref{sec:morecomm}. In this section we prove Theorem \ref{thm:more}. The proof
is divided in the following two lemmas. 

\begin{lemma} 
    \label{lem:morecom}
    Pick $\bx \sim \{-1,1\}^{kn}$ u.a.r. Then, with high probability, the vertices
    of $V_1$ are either all of positive type or all of negative type.
    Furthermore, the two events have equal probability.
\end{lemma}
\begin{proof}  We will write
\[ 
    \bx = \bx_{\bone} + \bx_{V_1} + \bx_{\perp_{1}} + \bx_{\perp},
\]
where $\bx_{\bone}$ is the component of $\bx$ parallel to $\bone$, $\bx_{V_1}$
is the component parallel to the vector $\bone_{V_1} - k^{-1}\bone_V$,
$\bx_{\perp_{1}}$ is the component in the eigenspace of $\lambda_2$ and
orthogonal to $\bone_{V_1} - k^{-1}\bone_V$, and $\bx_{\perp}$ is the component
orthogonal to $\bone$ and to the eigenspace of $\lambda_2$.

For the above the make sense, $\bone_{V_1} - k^{-1}\bone_V$ must be an eigenvector
of $\lambda_2$, which is easily verified because its entries sum to zero and
they are constant within components.

An important observation, and the reason for picking the above decomposition,
is that $\bx_{\perp_{1}}$ is zero in $V_1$. The reason is that
$\bx_{\perp_{1}}$ has to be orthogonal to $\bone_V$ and to $\bone_{V_1} -
k^{-1}\bone_V$ so from
\[ 
    \langle \bx_{\perp_{1}} , \bone_V \rangle =  \langle \bx_{\perp_{1}} ,
    \bone_{V_1} - k^{-1}\bone_V \rangle  = 0,
\]
we deduce
\[ 
    \langle \bx_{\perp_{1}} , \bone_{V_1}  \rangle = 0.
\]
Thus, the entries of $\bx_{\perp_{1}}$ sum to zero within $V_1$, but, being in
the eigenspace of $\lambda_2$, the entries of $\bx_{\perp_{1}}$ are constant
within components, and so they must be all zero within $V_1$.

Now we have
\[ 
    P^t \bx = \bx_{\bone} + \lambda_2^t \bx_{V_1} + \lambda_2^t  \bx_{\perp_{1}} + P^t \bx_{\perp},
\]
and so, for each $v\in V_1$ it holds 
\begin{equation} 
    \label{eq:diff} 
    (P^{t+1} \bx)_v - (P^t \bx)_v = \lambda_2^t \cdot (1-\lambda_2) (\bx_{V_1})_v + ((P^{t+1}- P^t) \bx_{\perp})_v.
\end{equation}
For $t > T$, the hypothesis $\lambda<(1-\epsilon)\lambda_2$ implies that  
\begin{equation}
    |(P^t \bx_{\perp})_v| \leq  || P^t \bx_{\perp}||_\infty \leq || P^t
    \bx_{\perp}|| \leq \lambda^t || \bx_{\perp}|| \leq \sqrt{n} \cdot
    \lambda^t \leq \frac 1 {n^{1.5}} \lambda_2^t.
    \label{eq:bound_on_x_perp}
\end{equation}
Moreover, for each $v\in V_1$ we have 
\begin{align}
    | (\bx_{V_1})_v| 
        &= \|\bone_{V_1} - k^{-1}\bone_{V}\|^{-2}
            \langle \bx, \bone_{V_1} - k^{-1}\bone_{V}\rangle\left( 1-k^{-1} \right) \\
        &= \frac{k}{(k-1)n} \left( \sum_{i\in V_1} x_i -\sum_{i \in V}\frac{x_i}{k} \right) 
            \left( \frac{k-1}{k} \right)
            = \frac{1}{n} \left( \sum_{i\in V_1} x_i -\sum_{i \in V}\frac{x_i}{k} \right),
    \label{eq:more_proj_val}
\end{align}
and 
\begin{align}
    || \bx_{V_1}||
        = \frac{\langle \bx, \bone_{V_1} - k^{-1}\bone_{V}\rangle}
                {\|\bone_{V_1} - k^{-1}\bone_{V}\|}
        = \sqrt{\frac{k}{(k-1)n}} \left( \sum_{i\in V_1} x_i -\sum_{i \in V}\frac{x_i}{k} \right),
    \label{eq:x_V_1_norm}
\end{align}
which imply that 
\begin{equation}
    | (\bx_{V_1})_v| = \sqrt{(1-1/k)/n}\, \|\bx_{V_1}\|.
    \label{eq:size_of_x_V_1}
\end{equation}
Finally, note that by Lemma \ref{lemma:apxrdm} it holds w.h.p. $|| \bx_{V_1} ||
\geq \frac 1{n} || \bx || \geq \sqrt{\nfrac k{n}}$. 

The latter fact together with \eqref{eq:bound_on_x_perp} and
\eqref{eq:size_of_x_V_1} imply that w.h.p. the sign of \eqref{eq:diff} is the
same as the sign of $(\bx_{V_1})_v$, which is the same for all elements of
$V_1$ and is equally likely to be positive or negative.
\end{proof}

Of course the same statement is true if we replace $V_1$ by $V_i$ for any
$i=1,\ldots, k$; by a union bound, it is also true for all $i$ simultaneously
with high probability. 

\begin{lemma} 
    Pick $\bx \sim \{-1,1\}^{kn}$ u.a.r. There is an absolute constant $p$ (e.g.,
    $p = \frac 1 {100}$) such that, with probability at least $p$, all vertices
    of $V_1$ have the same type, all vertices of $V_2$ have the same type, and
    the types are different.
\end{lemma}

\begin{proof} This time we write
    \[ \bx = \bx_{\bone} + \bx_{V_{1+2}} + \bx_{V_{1-2}} + \bx_{\perp_{1,2}} + \bx_{\perp} \]
where 
\begin{itemize}
    \item $\bx_{\bone}$ is the component parallel to $\bone_V$, 
    \item $\bx_{V_{1+2}}$ is the component parallel to $\bone_{V_1} +  \bone_{V_2} - \frac 2k \bone_V$, 
    \item  $\bx_{V_{1-2}}$ is the component parallel to $\bone_{V_1} - \bone_{V_2}$, 
    \item $\bx_{\perp_{1,2}}$ is the component in the eigenspace of $\lambda_2$
        and orthogonal to $\bx_{V_{1+2}}$ and $\bx_{V_{1-2}}$ , 
    \item $\bx_{\perp}$ is the rest.
\end{itemize}
Similarly to the proof of Lemma \ref{lem:morecom}, the important observations
are that $\bx_{V_{1+2}}$ and $\bx_{V_{1-2}}$  are in the eigenspace of
$\lambda_2$, and that $\bx_{\perp_{1,2}}$ is zero in all the coordinates of
$V_1$ and of $V_2$.

Thus, for each $v\in V_1 \cup V_2$ we have
\begin{equation} 
    \label{eq:diff.two} 
    (P^{t+1} \bx)_v - (P^t \bx)_v = \lambda_2^t (1-\lambda_2)
    (\bx_{V_{1+2}} + \bx_{V_{1-2}})_v + ((P^{t+1}- P^t) \bx_{\perp})_v.
\end{equation}
From \eqref{eq:diff.two} it is easy to see that if $\bx$ is such that, for
every $v\in V_1 \cup V_2$, we have the two conditions
\begin{align}
    | (\bx_{V_{1+2}})_v| &\leq \frac 34  | (\bx_{V_{1-2}})_v| \quad\text{and}
    \label{eq:more_precon}\\
    | ((P^{t+1}- P^t) \bx_{\perp})_v |  &\leq \frac 18 \lambda_2^t \cdot
        (1-\lambda_2) \cdot | (\bx_{V_{1-2}})_v |,
    \label{eq:morecomm_cond}
\end{align}
then such an $\bx$ satisfies the conditions of the Lemma, that is all the
elements in $V_1$ have the same type, all the elements of $V_2$ have the same
type, and the types are different. Now note that, since 
\begin{align}
    | (\bx_{V_{1+2}})_v|
    & = \frac{1}{2n}\left( \sum_{i\in V_1}x_i + \sum_{i\in
        V_1}x_i - \frac{2}{k}\sum_{i\in V}x_i\right) \quad\text{and}\\
    | (\bx_{V_{1-2}})_v|
    & = \frac{1}{2n}\left( \sum_{i\in V_1}x_i - \sum_{i\in V_2}x_i \right),
    \label{eq:more_proj}
\end{align}
if $\bx$ satisfies 
\begin{align}
    2 \sqrt{n} \leq \sum_{v\in V_1} &x_v \leq 3\sqrt{n},
    \label{eq:more_firstcond}\\
    - 2 \sqrt{n} \leq \sum_{v\in V_2} &x_v \leq - \sqrt{n} \quad\text{and}
    \label{eq:more_secondcond}\\
    0 \leq \sum_{v\in V / (V_1\cup V_2)} &x_v \leq \frac 1 {10}\sqrt{kn},
    \label{eq:more_thirdcond}
\end{align}
then \eqref{eq:more_precon} is satisfied, and note that \eqref{eq:more_firstcond}, 
\eqref{eq:more_secondcond} and \eqref{eq:more_thirdcond} are independent and
each happens with constant probability. 

Finally, observe that if \eqref{eq:more_precon} holds then \eqref{eq:morecomm_cond} is
satisfied with high probability when $t>T$. 
\end{proof}

It is enough to pick $\ell = \log (3n)$ to have, with high probability, that 
the signatures are well defined and they are the same within each community 
and different between communities. The first lemma guarantees that, with high 
probability, for all $\ell$ vectors, all vertices within each community have 
the same type. The second lemma guarantees that, with high probability, the 
signatures are different between communities.